\documentclass[a4paper]{article}

\usepackage{subfigure}
\usepackage{a4wide}
\usepackage{epic,eepic}
\usepackage{graphicx,color}

\usepackage{float}
\usepackage{ifthen}
\usepackage{xspace}

\usepackage{amssymb}
\usepackage{amsmath}
\usepackage{amsfonts} 
\usepackage{stmaryrd} 
\usepackage{latexsym} 

\newtheorem{theorem}{Theorem}

\newtheorem{lemma}{Lemma}
\newtheorem{corollary}{Corollary}
\newtheorem{definition}{Definition}
\newtheorem{remark}{Remark}

\newenvironment{proof}{\noindent\textbf{Proof:}}{\hfill $\blacksquare$\\}
\newcommand\set[1]{\ensuremath{\{ #1 \} }}
\renewcommand\int[1]{\ensuremath{\llbracket #1 \rrbracket}}
\newcommand\cur[1]{\ensuremath{{\cal{#1}}}}

\usepackage{url}

\begin{document}

\begin{center}

{\LARGE \bf On the Termination of Some Biclique Operators on Multipartite Graphs\,\footnote{This work was partially supported by the PICS program of CNRS (France), by the project ``Models on graphs: enumerative combinatorics and algebraic structures" of the Vietnam National Foundation for Science and Technology Development (NAFOSTED) and by the Vietnam Institute for Advanced Study in Mathematics (VIASM).
This article is a complete and improved version of the extended abstract that appeared in \cite{LPCN10}.}}

\bigskip

Christophe Crespelle\,\footnote{Corresponding author: christophe.crespelle@inria.fr}
\ \ \ \ \ \ 
Matthieu Latapy
\ \ \ \ \ \ 
Thi Ha Duong Phan

\bigskip
\bigskip


\begin{minipage}{.9\textwidth}
\centerline{\bf Abstract.}
\medskip
We define a new graph operator, called the \emph{weak-factor graph}, which comes from the context of complex network modelling. The weak-factor operator is close to the well-known clique-graph operator but it rather operates in terms of bicliques in a multipartite graph. We address the problem of the termination of the series of graphs obtained by iteratively applying the weak-factor operator starting from a given input graph. As for the clique-graph operator, it turns out that some graphs give rise to series that do not terminate. Therefore, we design a slight variation of the weak-factor operator, called \emph{clean-factor}, and prove that its associated series terminates for all input graphs. In addition, we show that the multipartite graph on which the series terminates has a very nice combinatorial structure: we exhibit a bijection between its vertices and the chains of the inclusion order on the intersections of the maximal cliques of the input graph.

\bigskip

{\bf Keywords.} Clean-factor Graph, Multipartite Graphs, Graph Series, Complex Network Modelling
\end{minipage}

\end{center}

\section{Introduction}\label{sec-intro}

The \emph{clique-graph} operator \cite{S03} is a well-known graph operator which, given a graph $G$, consists of building the graph $G'$ whose vertices are maximal cliques of $G$ and such that there is an edge between two distinct vertices of $G'$ iff the corresponding cliques of $G$ share at least one common vertex. The clique-graph series, obtained by iteratively applying the clique-graph operator starting from $G$, has been widely studied (see e.g. \cite{P92,BS95}). This series is said to be convergent (in the sense of~\cite{BS95}) if one of the graphs of the series is the graph with one single vertex\footnote{Note that \cite{P92} uses a different definition of convergence which includes the one of \cite{BS95} as a particular case, and also includes periodic behaviours. The notion of termination we use throughout the article is somehow equivalent to the one of \cite{BS95}.}. Then, all the graphs obtained in the following iterations are the same (i.e. reduced to a single vertex).

Here we consider a new operator, called \emph{weak-factor graph}, which comes from the context of complex network modelling and which operates in terms of bicliques in multipartite graphs rather than cliques in graphs. One of the interest of this operator is that it keeps an explicit and complete track of the original graph at all step of the series: given an arbitrary graph of the series it generates, one can univocally retrieve the original graph which gave rise to the series. Given a $k$-partite graph $G=(V_0,\ldots,V_{k-1},E)$ (see Section~\ref{sec:prel} for a definition), where $k\geq 2$, the \emph{weak-factor graph} $G'$ of $G$ is defined as follows (see Definition~\ref{weakfactor} for a more formal definition): $G'$ is the graph $G$ augmented with a new level $V_k$ of vertices, each of which corresponds to one \emph{non-simple} maximal biclique of the bipartite graph which is formed by the edges between the upper level of $G$, i.e. $V_{k-1}$, and the rest of the vertices of $G$, i.e. $\bigcup_{0\leq i\leq k-2} V_i$. Note that we do not add a new vertex at level $V_k$ for all maximal bicliques but only for the \emph{non-simple} ones, i.e. those maximal bicliques that have at least two vertices in $V_{k-1}$ and two vertices in $\bigcup_{0\leq i\leq k-2} V_i$. For each new vertex $x$ at level $V_k$ corresponding to a non-simple maximal biclique $B$, we define its neighbourhood in $G'$ as being the vertices of $B$. The weak-factor series of a graph $G$ is defined as the series obtained by iteratively applying the weak-factor operator starting from the vertex-clique-incidence bipartite graph of $G$ (see Figure~\ref{fig-multipartite}), where the vertices of $G$ are at level $V_0$ and the maximal cliques of $G$ are at level $V_1$.
This series is said to terminate iff, at some point, no new vertices are created. Then, the series is finite and the following graphs are undefined. For example, the series depicted on Figure~\ref{fig-multipartite} terminates since no new vertices are created when applying the weak-factor operator on graph $G_3$ of the series.


As we will show in Section~\ref{sec-wfs}, the weak-factor series does not always terminate. Then, the interest of the multipartite structure of the weak-factor graph is that it will allow us to restrict the definition of the non-simple bicliques we use in the weak-factor operator, taking into account the different levels of the multipartite graphs. In this way, we will be able to devise a refined version of the weak-factor graph, which we call the \emph{clean-factor graph}, whose associated series terminates for all graphs.


\begin{figure}[!h]
\centering
\input{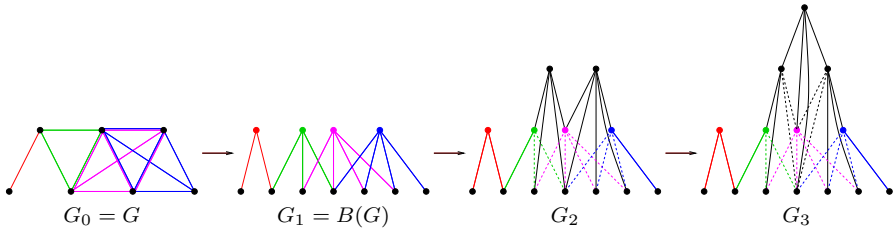}
\caption{Example of the weak-factor series of some graph $G$. From left to right: the original graph $G=G_0$, its vertex-clique-incidence bipartite graph $B(G)=G_1$, the tripartite graph $G_2$ of the series, and the quadripartite graph $G_3$ on which the series terminates. The dashed edges are those belonging to some non-simple maximal biclique used in the factorisation steps. Note that even if those edges were removed, the multipartite graph obtained at termination of the series would still constitute an unambiguous encoding of the original graph, as these edges are encoded by the presence of some vertices in the levels above them.}
\label{fig-multipartite}
\end{figure}

\subsubsection*{Our contribution}

In this paper, we study a new graph operator called the weak-factor operator which naturally arises in complex network modelling (see Section~\ref{sec:motiv}). We show that there are graphs for which the weak-factor series, obtained by iteratively applying the weak-factor operator, does not terminate. Therefore, our main contribution is to define a refinement of the operator, called the \emph{clean-factor graph}, whose series terminates for all graphs, thereby defining an object suitable for modelling purposes. The difference between our approach and the one followed by previous works is that we do not try to determine on which graphs the weak-factor series terminates, but we rather look for minimal constraints to impose to the operator in order to obtain termination for all graphs. The solution we present here obtains termination by imposing constraints to only a bounded number of levels (namely 3) of the multipartite graph on which the operator is applied.


In addition to our termination result, we show that the multipartite graph on which the clean-factor series terminates has a remarkable combinatorial structure. Namely, its vertices are in bijection with the chains of the inclusion order on the non-simple intersections of maximal cliques of the graph (Theorem~\ref{ThCharSeq}), denoted $\cur{L}$ in the rest of the paper. We believe that this link between the termination of the series and the structure of the cliques of the original graph is worth in itself and may be used to study termination of other similarly defined graph operators.\\
Finally, we give an upper bound on the size and computation time of the graph on which the iterated clean-factor series of $G$ terminates, under reasonable hypotheses on the degree distributions of the vertex-clique-incidence bipartite graph of $G$ (which hold for most real-world complex networks), therefore showing that this multipartite graph can be used in practice for complex network modelling.

Let us mention that this work is an improved and complete version of the extended abstract that appeared in \cite{LPCN10}. In \cite{LPCN10}, the notion of clean-factor is slightly different from the one we use here. As a consequence, we could not prove that imposing constraints only to a bounded number of levels is enough to guarantee the termination of the series, and we did not have a real bijection between the multipartite graph obtained at termination and the vertices of $\cur{L}$.

\subsubsection*{Related works}

The weak-factor operator we study here operates on multipartite graphs and is defined using the bicliques between the upper level and the rest of the multipartite graph. For graphs, closely related operators have been defined using the cliques or the bicliques of the graph, and many works addressed the question of convergence of the series obtained by iteratively applying these operators to an input graph. There exists several definitions of convergence in the literature. The notion of termination we use here for the multipartite graph series is somehow equivalent to the convergence notion used in~\cite{BS95} in the context of graph series, and is a particular case of convergence of the definition used in~\cite{P92}.

For the well-known clique graph operator (see \cite{S03} for a survey) the question of convergence has received a lot of attention \cite{P92,BS95}. Most of the efforts focussed on obtaining convergence results, or divergence results, for some particular graphs or graph classes \cite{LMMNP04,LNP04,LPV08,LSS10}. Similar questions have been addressed recently for the biclique graph operator \cite{GM09,GM13}, which also operates on graphs\footnote{Note that the bicliques we use in the definition of the weak-factor operator are bicliques in a bipartite graph and not in a general graph as in the definition of the biclique graph operator.} but using bicliques instead of cliques. It is worth noticing that the clique graph and the biclique graph are defined as intersection graphs while this is not the case of the operators we study in this paper.
Let us mention, that another closely related graph operator called edge-clique-graph operator has been studied (see e.g. \cite{CKMTS91,C03}) but, to the best of our knowledge, the question of the convergence of its iterated series has not been investigated.

It must be clear that none of these three operators, clique graphs, biclique graphs and edge-clique graphs, which are defined on graphs, is equivalent to one of the multipartite-graph operators we consider here. And the convergence or divergence results obtained previously for these graph operators do not imply the termination and non-termination results we present here. 

Moreover, it is worth noticing that the question we address in this paper is orthogonal, and complementary, to the one addressed in all the previously cited works. Indeed, we do not intend to characterise the graphs for which the iteration of the weak-factor operator terminates or does not terminate. Instead, we aim at determining minimal constraints that can be imposed to this operator in order to obtain termination for all graphs.

Finally, we note that recently \cite{E10} showed the interest of clique graphs to study communities in complex networks. However, their approach and results are not equivalent to ours. In particular, they do not consider the series obtained by iterating the operator, which is our main concern here in the case of the weak-factor operator.

\subsubsection*{Outline of the paper}

In Section~\ref{sec:motiv}, we detail further the context where the practical motivation of our work comes from. Then, Section~\ref{sec:prel} gives a few notations and basic definitions, useful in the whole paper, including the definition of a fundamental notion, the {\em factorisation}, which plays a key role in the following.
In Section~\ref{sec-wfs}, we formally define the weak-factor operator and a natural variation of this operator, called the factor operator, both of which giving rise to some infinite series. In Section~\ref{sec-cfs} we propose a deeper refinement of the operator, called the clean-factor operator, for which we prove termination for all graphs and give a structural characterisation of the multipartite graph obtained at termination. Finally, in Section~\ref{sec-pract}, we address the question of efficiently computing and storing the representation provided by the clean-factor series.

\subsection{Motivation}\label{sec:motiv}

It is worth to mention that we did not come to the study of the termination of the weak-factor series only for theoretic motivations: this question is of key interest in complex network modelling. Complex networks are those graphs encountered in practice in various domains such as computer science, biology, social sciences and others. In the last decade, they were shown to share some nontrivial common properties \cite{watts98collective,albert02statistical}, independently from the context they come from. A lot of efforts have been done to design models able to capture these properties while staying general enough. One of the difficulty of the domain is to encompass in a same model the two major properties of these networks, namely their heterogeneous degree distribution and their high local density (clustering coefficient, see \cite{barrat2010} for a formal definition).

Among the most promising approaches, \cite{ipl04guillaume,physicaa06guillaume} propose to model complex networks based on the properties of their vertex-clique-incidence bipartite graph. Their idea is to use prescribed-degree-graph generation, which is a powerful and well understood technique since the works of \cite{BC78,molloy95critical}, for the vertex-clique-incidence bipartite graph instead of the graph itself. In other words, they advocate for the generation of complex networks by their cliques rather than by their edges. They show that, in this way, one obtains graphs having a high local density (thanks to the clique structure) and a heterogeneous degree distribution that is controlled by the degrees of the vertices in the vertex-clique-incidence bipartite graph. However, the bipartite model suffers from a severe limitation: when generating the edges of the bipartite graph at random, the obtained neighbourhoods of the upper vertices intersect only on one (or zero) vertex with a very high probability (see \cite{ipl04guillaume,physicaa06guillaume}). This is not the case in real world networks, where most of the maximal cliques have non-simple overlaps with some others (i.e. overlaps of cardinality at least two). Thus, even though it gives the desired properties concerning degree distribution and local density, the bipartite model results in graphs having a caricaturistic structure.

The weak-factor graph (see Section~\ref{sec-wfs}) was introduced in \cite{LPCN10} in order to correct this drawback. The idea is to define an object that encodes the non-simple intersections of maximal cliques of a graph $G$ by the neighbourhoods of vertices in some other suitably defined graph, so that such objects can be randomly generated using the prescribed-degree generation technique of \cite{BC78,molloy95critical}. In order to define such an encoding of a graph $G$, we can proceed as follows. We start from the vertex-clique-incidence bipartite graph $B(G)=(V_0,V_1,E)$ of $G$ and we create a new level $V_2$ where each vertex $x$ corresponds to a non-simple maximal biclique $B$ of $B(G)$. Then, we can delete the edges of $B$ as they are now encoded by the presence of $x$. Doing so simultaneously for all non-simple maximal bicliques of $B(G)$ gives a tripartite graph in which the neighbourhoods on $V_0$ of vertices at level $V_1$ have no non-simple intersections anymore. Then, we can iteratively repeat the operation by considering, at each stage of the process, the maximal bicliques between the vertices on the uppermost level and the rest of the vertices of the multipartite graph, until the process hopefully terminates. In this case, we obtain a multipartite graph\footnote{Note that this multipartite graph is an encoding of the original graph, as the factorising operation is reversible.} without any non-simple intersection of neighbourhoods. We can therefore generate similar structures at random using the prescribed-degree generation method without bumping into the problem raised by \cite{ipl04guillaume,physicaa06guillaume}. Of course, in order to obtain a multipartite graph that has no non-simple neighbourhood intersections, it is mandatory that the iterative factorising process terminates. This is the reason why we came to study the termination of the weak-factor series.

Note that, opposite to the process described above, in the definition of the weak-factor operator we do not delete edges of the bicliques involved in one factorisation step. This has no impact on the set of nodes created in the rest of the process, as those edges are not involved in further factorisation steps. On the other hand, keeping those edges helps to describe the structure of the graphs of the weak-factor series (Definition~\ref{def:factorisation} below) and this is the reason why we keep them in the rest of the paper.

\subsection{Notations and preliminary definitions}\label{sec:prel}

All graphs considered here are finite, undirected and simple (no loops and no multiple edges). A graph $G$ having vertex set $V$ and edge set $E$ will be denoted by $G=(V,E)$. We also denote by $V(G)$ the vertex set of $G$. The edge between vertices $x$ and $y$ will be indifferently denoted by $xy$ or $yx$. $\cur{K}(G)$ denotes the set of maximal cliques of a graph $G$, and $N(x)$ the neighbourhood of a vertex $x$ in $G$.

An ordered $k$-partition $\cur{P}$ of a set $V$ is a partition of $V$ into $k$ parts (non empty and pairwise disjoint, from the classical definition of partition) which are numbered from $0$ to $k-1$. It is denoted as a $k$-tuple: $\cur{P}=(V_0,\ldots,V_{k-1})$. In this paper, a $k$-partite graph is always given together with a partition of its vertices as in the following definition.

\begin{definition}[$k$-partite graph] A \emph{$k$-partite graph} is a couple $(G,\cur{P})$ where $G=(V,E)$ is a graph and $\cur{P}=(V_0,\ldots,V_{k-1})$ is an ordered $k$-partition of its vertex set $V$ such that all edges of $G$ are between vertices in different parts of $\cur{P}$. It is denoted by $G=(V_0,\ldots,V_{k-1},E)$.\\
\end{definition}

A \emph{multipartite} graph is a $k$-partite graph with $k\geq 2$. For a $k$-partite graph $G=(V_0,\ldots,V_{k-1},E)$, the vertices of $V_i$, for any $i$, are called the {\em $i$-th level} of $G$, and the vertices of $V_{k-1}$ are called its \emph{upper vertices}. We denote by $N_i(x)$, where $0\leq i\leq k-1$, the set of neighbours of $x$ at level $i$: $N_i(x)=N(x)\cap V_i$. 
A \emph{biclique} of a graph is a set of vertices of the graph inducing a complete bipartite graph.
We denote by $B(G)$ the vertex-clique-incidence bipartite graph of $G=(V,E)$: $B(G)=(V,\cur{K}(G),E')$ where $E' = \{vc \ |\ c \in \cur{K}(G), \ v \in c\}$.
A \emph{non-simple biclique} of a bipartite graph is a biclique having at least two vertices in the upper level and at least two vertices in the bottom level.
Two sets have a \emph{non-simple intersection} if they share at least two elements.
In the whole paper, we denote $\cur{L}$ the inclusion order of the non-simple intersections of maximal cliques of a graph $G$ (there will be no confusion on the graph $G$ referred to when we use this notation).

For two non-negative integers $a,b\in\mathbb{N}$, we use the notation $\int{a,b}$ for the set $\set{p\in\mathbb{N}\ |\ a\leq p\leq b}$, with the convention $\int{a,b}=\varnothing$ if $a>b$.

In the sequel, an operation will play a key role, we name it \emph{factorisation} and define it generically as follows.

\begin{definition}[factorisation with respect to $V'_k(G)$]\label{def:factorisation}
Given a $k$-partite graph $G=(V_0,\ldots,V_{k-1},E)$ with $k\geq 2$ and a set $V'_k(G)$ of subsets of $V(G)$, we define the factorisation of $G$ with respect to $V'_k(G)$ as the $(k+1)$-partite graph $G'=(V_0,\ldots,V_k,E\cup E_+)$ where:
\begin{itemize}
\item $V_k$ is the set of maximal (with respect to inclusion) elements of $V'_{k}(G)$,
\item $E_+=\set{Xy\ |\ X\in V_{k} \text{ and } y\in X}$.
\end{itemize}
When $V_k\not=\emptyset$, the factorisation is said to be \emph{effective}.
\end{definition}

Provided that the set $V'_k(G)$ is properly defined for all multipartite graphs $G$, such a factorisation operation defines a multipartite graph operator, the iteration of which gives rise to a series of multipartite graphs as defined below.

\begin{definition}[series associated to a factorisation operation]\label{def:series}
Given a factorisation operation that associates any $k$-partite graph $G=(V_0,\ldots,V_{k-1},E)$ with $k\geq 2$ to a (unique) $k+1$-partite graph $G'$
, we define the series of multipartite graphs $(G_i)_{i\geq 1}$, associated to this factorisation operation and generated by a graph $G_0=(V_0,E_0)$, by: $G_1=B(G_0)$ is the vertex-clique-incidence bipartite graph of $G_0$ (in which the cliques are on the upper level of $B(G_0)$) and, for all $i\geq 1$, $G_{i+1}= G'_{i}$ when the factorisation of $G_i$ is effective, and $G_{i+1}$ is undefined otherwise.
\end{definition}

\begin{definition}[termination of the series]\label{def:conv}
We say that the series $(G_i)_{1\leq i\leq n}$ associated to some factorisation operation \emph{terminates} iff for some $i\geq 1$ the factorisation is not effective, then all subsequent graphs of the series are undefined and the series reduces to a finite sequence.
\end{definition}

In the rest of the paper, we will refine the notion of factorisation by using different sets $V'_k(G)$ on which is based the factorisation operation. And we will study termination of the graph series resulting from each of these refinements. As, in the following, the graph $G$ referred to is always clear from the context, we denote $V'_k$ instead of $V'_k(G)$. But all the sets $V'_k$ we define still depend on the graph considered.

\section{Weak-factor series and factor series}\label{sec-wfs}

\subsection{Weak-factor series}

Thanks to the generic notion of factorisation, we will now formally define the weak-factor operation we introduced above.

\begin{definition}[$V^+_{k}$ and weak-factor graph]\label{weakfactor}
Given a $k$-partite graph $G=(V_0,\ldots,V_{k-1},E)$ with $k\geq 2$, we define the set $V^+_{k}$ as:
$$V^+_{k}=\set{\set{x_1,\ldots,x_l}\cup\bigcap_{1\leq i\leq l} N(x_i)\ |\ l \ge 2,\ \forall i\in\int{1,l}, x_i\in V_{k-1} \text{ and } |\bigcap_{1\leq i\leq l} N(x_i)|\geq 2}.$$
The \emph{weak-factor graph} $G^+$ of $G$ is the factorisation of $G$ with respect to $V^+_{k}$.
\end{definition}


\begin{figure}[!h]
\centering
\includegraphics[scale=0.5]{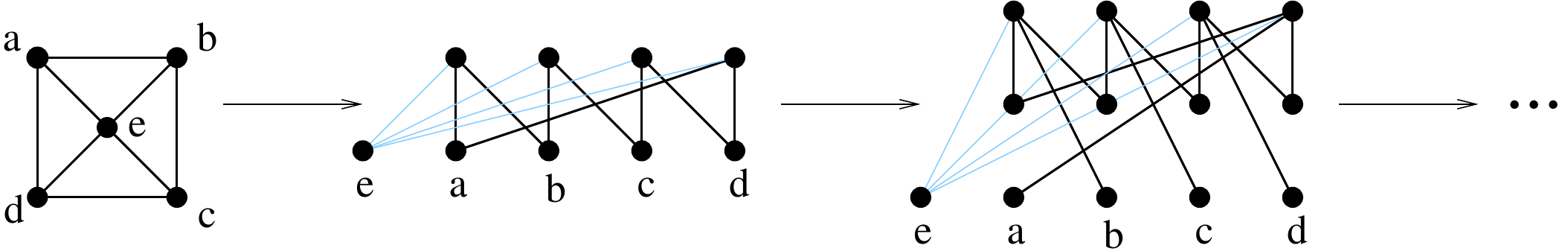}
\caption{An example graph for which the weak-factor series is infinite. From left to right: the input graph $G$, its vertex-clique-incidence bipartite graph $B(G)$, and the tripartite graph $B(G)^+$ of its weak-factor series (note that for sake of readability of the drawing, the edges between levels $V_1$ and $V_0$ have been removed in $B(G)^+$). The shaded edges are the ones involving vertex $e$, which plays a special role: all the vertices of the upper level of the multipartite graph are linked to $e$. The structure of edges between vertices of $V_2$ and vertices of $V_1\cup\set{e}$ in $B(G)^+$ is identical to the one between levels $V_1$ and $V_0$ in $B(G)$, revealing that the series will not terminate.}
\label{fig-infinite}
\end{figure}

Figure~\ref{fig-multipartite} gives an illustration for this definition. In this case, the weak-factor series is finite. However, it is not difficult to find examples of graphs which lead to infinite weak-factor series. Figure~\ref{fig-infinite} provides such an example for which the structure of the upper level is infinitely reproduced on the further levels. Intuitively, this is due to the fact that a vertex may be the base for an infinite number of factorisation steps (vertex $e$ in the example of Figure~\ref{fig-infinite}). The aim of the next sections is to avoid this case by using a more restrictive definition of factorisation.

\subsection{Factor series}

In this section, we examine a first restriction of the operator, called \emph{factor graph}, that forbids the repeated use of the same vertex to produce infinitely many factorisations, which is the phenomenon responsible for the non termination of the series for the example of Figure~\ref{fig-infinite}.

\begin{definition}[$V^\circ_{k}$ and factor graph]\label{factor}
Given a $k$-partite graph $G=(V_0,\ldots,V_{k-1},E)$ with $k\geq 2$, we define the set $V^\circ_{k}$ as:
$$V^\circ_{k}\ =\ \set{X\in V^+_k \text{ such that } |\bigcap_{y\in X\cap V_{k-1}} N_{k-2}(y)| \geq 2}.$$
The \emph{factor graph} $G^\circ$ of $G$ is the factorisation of $G$ with respect to $V^\circ_{k}$.
\end{definition}

This new definition results from the restriction of the weak-factor definition by considering only sets $X\in V^+_k$ such that the vertices of $X\cap V_{k-1}$ have at least two common neighbours at level $k-2$. The reason is that, in this way, a vertex cannot contribute to more than two factorising steps: once when it is on the upper level of the multipartite graph, once when it is on the level just below. Indeed, even if the vertices of levels lower than the two upper levels (i.e. $V_k$ and $V_{k-1}$) may be involved in a factorisation step, they are not responsible for the creation of a new vertex. Such a creation depends only on the edges between the two upper levels of the multipartite graph.

Finding examples of input graphs that generate infinite factor series is not straightforward. In particular, one natural candidate that one could have in mind, namely the graph whose vertex-clique-incidence bipartite graph is the \emph{anti-matching} on $2n$ vertices, which generates an infinite series for the clique-graph operator~\cite{Neu78}, actually gives rise to a finite series for the factor operator. The anti-matching is the bipartite complement of a perfect matching between the $n$ upper vertices and the $n$ bottom vertices (also known as the \emph{octahedron} in clique-graph theory and as the \emph{crown} in poset theory). The anti-matching is known for being the bipartite graph on $2n$ vertices that has the maximum number of bicliques. Regarding the factor series, it implies that there is a combinatorial explosion of the number of vertices on the first next levels. Despite of this, one can check that the series of the anti-matching on $2n$ vertices terminates.

Nevertheless, \cite{CPT15} recently provided an example of a graph that gives rise to an infinite factor series. This is the reason why in the next section, we constraint further the factor operation: we do not only require that the neighbourhoods of vertices at level $V_{k-1}$ involved in the creation of a new vertex at level $V_k$ share at least two vertices on level $V_{k-2}$ but we also require that those vertices have the same neighbourhood at level $V_{k-3}$ (see Definition~\ref{CleanFactor} of the clean-factor graph). This supplementary condition is not only a technical condition used to guarantee termination: we will show that the graph on which terminates the clean-factor series is a fundamental combinatorial object.

\section{Clean-factor series}\label{sec-cfs}

In the previous section, we studied two series of multipartite graphs based on two different factorisation operations, both of them giving rise to some infinite series. In this section, we introduce a more constrained refinement of these two factorisation operations, which we call the \emph{clean-factor} operator, for which we prove that the associated series always terminates. One interesting point of our solution is that the constraints introduced in order to guarantee termination are light: they apply on only 3 levels of the multipartite graph. In addition, the multipartite graph obtained at termination has a very deep combinatorial meaning: its vertices are the chains of the inclusion order $\cur{L}$ of the non-simple intersections of maximal cliques of the input graph $G$.

We now give the formal definition of the clean-factor graph of a multipartite graph: the general factorisation step is the case where $k\geq 5$, the construction of levels $V_2, V_3$ and $V_4$ are subject to particular conditions. It should be clear that these particular conditions may be simplified while preserving termination. But on the other hand, those exact conditions are necessary in order to obtain the bijection with the chains of order $\cur{L}$.

\begin{definition}[$V^*_{k}$ and clean-factor graph]\label{CleanFactor}
Given a $k$-partite graph $G=(V_0,\ldots,V_{k-1},E)$ with $k\geq 2$, we define the set $V^*_{k}$ as:
\begin{itemize}
\item If $k\geq 5$, $V^*_{k}\ =\ \set{X\in V^+_k\ |\ |\bigcap_{x\in X\cap V_{k-1}}N_{k-2}(x)|\geq 2 \text{ and }\forall x,y\in X\cap V_{k-1}, N_{k-3}(x)=N_{k-3}(y) \text{ and } |\bigcap_{x\in X\cap V_{k-1}}N_1(x)|\geq 2}$.

\item If $k=4$, $V^*_{4}\ =\ \set{X\in V^+_4\ |\ |\bigcap_{x\in X\cap V_{3}}N_2(x)|\geq 2 \text{ and } |\bigcap_{x\in X\cap V_{3}}N_1(x)|\geq 2 \text{ and }\forall x,y\in X\cap V_{3}, N_0(x) = N_0(y)}$.

\item If $k=3$, $V^*_3=\ \set{X\in V^+_3\ |\ |\bigcap_{x\in X\cap V_2}N_1(x)|\geq 2 \text{ and } |\bigcap_{x\in X\cap V_2}N_0(x)|\geq 2}$.

\item If $k=2$, $V^*_2=V^+_2$.


\end{itemize}

The \emph{clean-factor graph} $G^*$ of $G$ is the factorisation of $G$ with respect to $V^*_{k}$.
\end{definition}

The rest of this section is devoted to proving the termination of the clean-factor series $(G_i)_{i\geq 1}$ generated by any graph $G$ (Theorem~\ref{CFSstop}) and the bijection between vertices of level $V_i$ of the series, with $i\geq 2$, and the chains of length $i-2$ of $\cur{L}$ (Theorem~\ref{ThCharSeq}).
We start by proving Theorem~\ref{ThCharSeq} since Theorem~\ref{CFSstop} will be obtained as a direct corollary from it.

Theorem~\ref{ThCharSeq} gives a characterisation of $V_i, i\geq 2$ by associating to each of its nodes a chain of length $i-2$ in order $\cur{L}$.
Formally, we associate to a node $x$ of $V_i$ a sequence $S(x)$ of subsets of $V(G)$ which are precisely the elements of $\cur{L}$ defining the chain associated to $x$. Before formally defining $S(x)$ (Definition~\ref{DefCharSeq}) and stating Theorem~\ref{ThCharSeq}, we need to establish some basic definitions, notations and properties of the non-simple intersections of the maximal cliques of a graph.

\begin{definition}\label{DefOK}
We denote by $\cur{O}'$ the set of intersections of maximal cliques of $G$ (possibly only one clique or none), that is $\cur{O}'=\set{O\subseteq V(G)\ |\ \exists P_c\subseteq \cur{K}(G), O=\bigcap_{C\in P_c} C}$, using the convention that $\bigcap_{C\in\varnothing} C=V(G)$. And we denote by $\cur{O}$ the subset of $\cur{O}'$ formed by the elements that contain at least two vertices of $G$ and that are obtained as the intersection of at least two distinct maximal cliques of $G$, that is $\cur{O}=\set{O\in\cur{O}'\ |\ |O|\geq 2 \text{ and } \exists k\geq 2, \exists C_1,\ldots ,C_k\in \cur{K}(G), (\forall j,l\in\int{1,k}, j\neq l \Rightarrow C_j\neq C_l) \text{ and } O=\bigcap_{1\leq i\leq k} C_i}$.
\end{definition}

\begin{definition}\label{defKC}
For any subset $A\subseteq V(G)$ of vertices of $G$, we denote by $K(A)$ the set of maximal cliques of $G$ containing $A$, that is $K(A)=\set{C\in \cur{K}(G)\ |\ A\subseteq C}$.
And we denote by $\cur{C}$ the family of subsets of $\cur{K}(G)$ defined by $\cur{C}=\set{K(O)\ |\ O\in\cur{O}'}$.
\end{definition}


Note that the set $A$ defining $K(A)$ is not unique: there may exist $A'\neq A$ such that $K(A')=K(A)$. This is the reason why we now need to state some basic properties of sets $K(A)$ that we will use in the following.

\begin{remark}\label{Kincl}
For any subsets $A,B\subseteq V(G)$, if $A\subseteq B$ then $K(B)\subseteq K(A)$. And for any subsets $A\subseteq V(G)$ and $O\in\cur{O}'$, if $K(O)\subseteq K(A)$ then $A\subseteq O$.
\end{remark}

\begin{proof}
The first part of the remark is self-evident. For the second part, note that, by definition of $\cur{O}'$, $O=\bigcap_{C\in K(O)} C$. And on the other hand, we have $A\subseteq \bigcap_{C\in K(A)} C\subseteq \bigcap_{C\in K(O)} C = O$.
\end{proof}

\begin{remark}\label{PteOK}
For any $A,B\subseteq V(G)$, $K(A)\cap K(B)=K(A\cup B)$.
Conversely, if $A_1,\ldots,A_n\subseteq V(G)$, with $n\geq 2$, and if $O\in\cur{O}'$ and if $\bigcap_{1\leq i\leq n} K(A_i)=K(O)$, then $\bigcup_{1\leq i\leq n} A_i\subseteq O$.
\end{remark}

\begin{proof}
Let $A,B\in V(G)$.
The cliques in $K(A)\cap K(B)$ are exactly the cliques that contain both $A$ and $B$, {\em i.e.} the cliques that contain $A\cup B$. Therefore $K(A)\cap K(B)=K(A\cup B)$.

Let $A_1,\ldots,A_n\subseteq V(G)$, with $n\geq 2$, and let $O\in\cur{O}'$ such that $\bigcap_{1\leq i\leq n} K(A_i)=K(O)$. From what precedes, $\bigcap_{1\leq i\leq n} K(A_i)=K(\bigcup_{1\leq i\leq n} A_i)$. Consequently, we have $K(O)\subseteq K(\bigcup_{1\leq i\leq n} A_i)$.
And since $O\in\cur{O}'$, from Remark~\ref{Kincl}, we have $\bigcup_{1\leq i\leq n} A_i\subseteq O$.
\end{proof}

\begin{lemma}
$\cur{O}'$ and $\cur{C}$ are closed under intersection.
\end{lemma}

\begin{proof}
The fact that $\cur{O}'$ is closed under intersection is clear from the definition. Let us show that $\cur{C}$ is closed under intersection. Let $k\geq 2$ and let $O_1,\ldots ,O_k\in \cur{O}'$. We prove that $\bigcap_{1\leq i\leq k} K(O_i) \in \cur{C}$. For that purpose, consider a set $O\in \cur{O}'$ such that $O\supseteq \bigcup_{1\leq i\leq k} O_i$ and which is minimal under inclusion. We will show that $\bigcap_{1\leq i\leq k} K(O_i)=K(O)$. Since $O\supseteq \bigcup_{1\leq i\leq k} O_i$, we have $K(O)\subseteq K(\bigcup_{1\leq i\leq k} O_i)=\bigcap_{1\leq i\leq k} K(O_i)$, from Remark~\ref{PteOK}. For the converse inclusion, consider $C\in K(\bigcup_{1\leq i\leq k} O_i)$. By definition, $C\in\cur{O}'$ and $\bigcup_{1\leq i\leq k} O_i\subseteq C$. Since $\bigcup_{1\leq i\leq k} O_i\subseteq O$, we also have $\bigcup_{1\leq i\leq k} O_i\subseteq O\cap C$. And since $O\cap C\in \cur{O}'$ (as $\cur{O}'$ is closed under intersection), the minimality of $O$ implies that $O\cap C=O$, that is $O\subseteq C$. Thus, $C\in K(O)$ and it follows that $K(\bigcup_{1\leq i\leq k} O_i)\subseteq K(O)$, which completes the proof.
\end{proof}

The following lemma is the first step toward the bijection theorem (Theorem~\ref{ThCharSeq}). It establishes the bijection between vertices of $V_2$ and the chains of length $0$ of $\cur{L}$. We will use it in the initialising step of the recursion of the proof of Theorem~\ref{ThCharSeq}.

\begin{lemma}\label{Levels012}
In the clean-factor series, $V_0= V(G), V_1 = \cur{K}(G)$ and $V_2 = \cur{O}$ in the sense that the map $\phi$ defined by $x \mapsto N_0(x)$ is a bijection from $V_2$ to $\cur{O}$. Moreover, $\forall x\in V_2, N_1(x)=K(N_0(x))$. 
\end{lemma}

\begin{proof}
Let us start with the second part of the lemma. Let $x\in V_2$. By definition of $V^+_2$, all the elements $y$ in $N_1(x)$ are such that $N_0(x)\subseteq N_0(y)$. Then, $y\in K(N_0(x))$, by identifying $V_1$ and $\cur{K}(G)$. On the other hand, the maximality of $x$ in $V^+_2$ implies that all $y\in K(N_0(x))$ belong to $N_1(x)$. Thus, $N_1(x)=K(N_0(x))$

Let us prove that the map $\phi$ is a bijection from $V_2$ to $\cur{O}$.
First, if $x \in V_2$ then by definition, $|N_0(x)|\geq 2$, $|N_1(x)|\geq 2$, and $N_0(x)=\bigcap_{y\in N_1(x)} N_0(y)$, hence $N_0(x)$ belongs to $\cur{O}$, and the map $\phi$ is well defined. 

Second, $\phi(x) = \phi(x')$ means $N_0(x) = N_0(x')$. But $N_1(x)$ is the set of all maximal cliques containing $N_0(x)$, and $N_1(x')$ is the same. Then, if $\phi(x) = \phi(x')$, we have $x=x'$: $\phi$ is injective.

We now prove that $\phi$ is surjective. Let $O$ be an element of $\cur{O}$, we show that the element $x = K(O)\cup\bigcap_{y\in K(O)} N_0(y)$ is an element of $V_2$ and $\phi(x)=O$. It is clear that $x \cap V_0 = O$, so $|x \cap V_0| \geq 2$. Since $O\in\cur{O}$, $|K(O)|\geq 2$, and we have $|x\cap V_1|\geq 2$. Then $x\in V^+_2$. Moreover, by definition, $K(O)$ is exactly the set of all maximal cliques containing $O$, then $x$ is maximal in $V^+_2$. It follows that $x$ is an element of $V_2$, and $\phi(x)=O$.
\end{proof}


We are now ready to give the definition of the sequence $S(x)$ that we associate to a vertex $x\in V_i$ with $i\geq 2$.

\begin{definition}[Characterising sequence $S(x)$]\label{DefCharSeq}
Let $G$ be a graph and let $(G_i)_{i\geq 1}$ be its clean-factor series.
The \emph{characterising sequence} $S(x)=(O_1(x),\ldots ,O_{k-1}(x))$ of a vertex $x\in V_k$, with $k\geq 2$, is defined by:
\begin{itemize}
\item $O_1(x)=N_0(x)$, and
\item for $k\geq 3$, $\forall j\in\int{2,k-1}, O_j(x)$ is the unique element of $\cur{O'}$ such that $K(O_j(x))=\bigcap_{y\in N_j(x)} N_1(y)$.
\end{itemize}
\end{definition}

Note that $O_j$ is properly defined. Indeed, from Lemma~\ref{Levels012}, $\forall y\in V_2, V_1(y)\in \cur{C}$. And since $\cur{C}$ is closed under intersection, a simple recursion shows that for all $i\geq 3$ and for all $y\in V_i $, $N_1(y)=\bigcap_{z\in V_{i-1}} N_1(z)\in \cur{C}$. Then, for any $j\geq 2$, $\bigcap_{y\in N_j(x)} N_1(y)$ is in $\cur{C}$ and there exists some $O_j$ in $\cur{O'}$ satisfying the condition. The fact that such an $O_j\in\cur{O}'$ is unique comes from the fact that for any set $O\in\cur{O}'$, we have $O=\bigcap_{C\in K(O)} C$. Then, if there exists some $O$ such that $K(O)=K(O_j)$, necessarily $O=O_j$. Consequently, $O_j$ is unique and properly defined.

We will often use the following remark in the proof of Theorem~\ref{ThCharSeq}.

\begin{remark}\label{V1}
For any $x\in V_k$, with $k\geq 2$, $K(O_{k-1}(x))=N_1(x)$.
\end{remark}

\begin{proof}
For $k=2$, the remark rewrites $K(O_1(x))=N_1(x)$. Since $O_1(x)=N_0(x)$ and since, from Lemma~\ref{Levels012}, $K(N_0(x))=N_1(x)$, then the result follows.
For $k>2$, the remark simply follows from the fact that $\bigcap_{y\in N_{k-1}(x)} N_1(y)=N_1(x)$.
\end{proof}

\begin{figure}[!h]
\centering
\input{g_and_l.tex_source}
\includegraphics[scale=0.49]{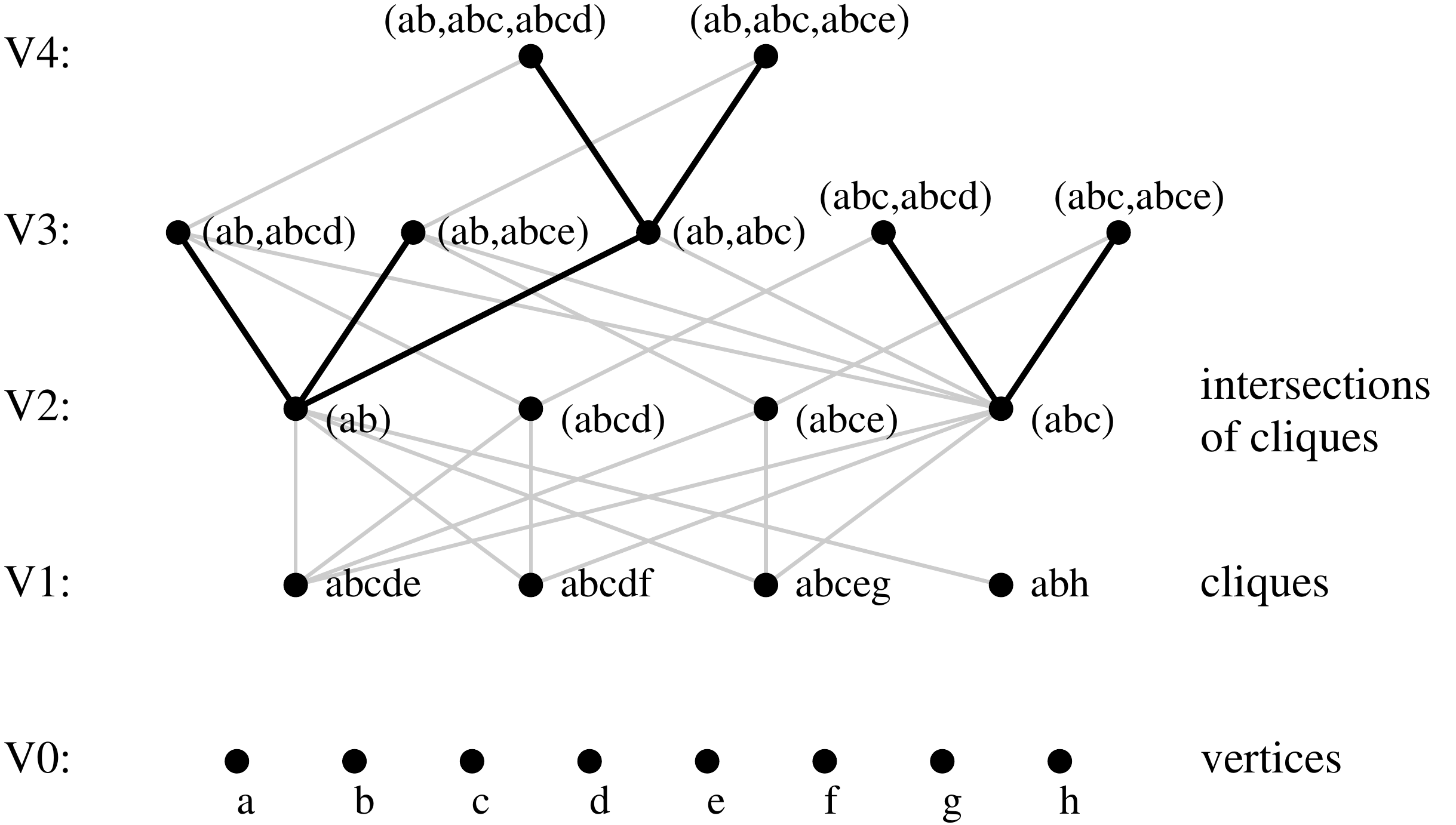}
\caption{Top left: a graph G. Top center: its maximal cliques. Top right: the inclusion order $\cur{L}$ of the non-simple intersections of maximal cliques of $G$. Bottom: the multipartite graph $M$ obtained at termination of the clean-factor series of $G$. $M$ has $5$ levels. Level $V_0$ is for the vertices of $G$ and level $V_1$ for the maximal cliques of $G$. The bijection between the vertices of $M$ and the chains of order $\cur{L}$ appears in the rest of the levels: each vertex $x$ of $M$ in levels $V_2$ to $V_4$ has been labelled with the corresponding chain of $\cur{L}$, i.e. its characterising sequence $S(x)$ (see Definition~\ref{DefCharSeq}). Level $V_2$ is for the non-simple intersections of maximal cliques of $G$ (i.e. the chains of $\cur{L}$ of length $0$), level $V_3$ is for the chains of $\cur{L}$ of length $1$, and level $V_4$ is for the chains of $\cur{L}$ of length $3$, which is precisely the height of $\cur{L}$. For sake of clarity, only a few links of $M$ have been drawn on the figure, namely all the links between two consecutive levels higher than level $V_2$. The black bold lines are for the link between a vertex $x$ at level $V_i$, for $i\geq 3$, and the unique vertex of $V_{i-1}$ whose characterising sequence is a prefix of the sequence of $x$ (of length $length(S(x))-1$, necessarily). In this way, we obtain a clearer representation of the bijection with the chains of $\cur{L}$, as one can clearly see the four prefix trees of the chains starting at a given element of $\cur{L}$: these prefix trees are rooted at level $V_2$ and two of them are reduced to one single node, labelled $(abcd)$ and $(abce)$.}
\label{fig-clean}
\end{figure}

We will now state the bijection theorem (Theorem~\ref{ThCharSeq}) which is our main combinatorial tool for proving the termination of the clean-factor series (Theorem~\ref{CFSstop}). Its proof is rather intricate, but it gives much more information than the termination of the series. By associating a sequence of sets to each vertex in levels greater than $V_2$ in the multipartite graph, we show that each such vertex corresponds to a chain of the inclusion order $\cur{L}$ of the non-simple intersections of maximal cliques of $G$.

\begin{theorem}[Bijection theorem]
\label{ThCharSeq}
Let $G$ be a graph and $(G_i)_{i\geq 2}$ its clean-factor series. For any $k\geq 2$, the map $\phi$ defined by $x \mapsto S(x)$ is a bijection from $V_k$ to $\set{(O_1,\ldots,O_{k-1})\in\cur{O}^{k-1}\ |\ O_1\subsetneq\ldots\subsetneq O_{k-1}}$ (see Figure~\ref{fig-clean}).
\end{theorem}
\begin{proof}

The case $k=2$ directly follows from Lemma~\ref{Levels012}, then, in the following, we only deal with the cases where $k\geq 3$. We prove Theorem~\ref{ThCharSeq} by recursion, using the five recursion hypotheses below, namely $H_{tar},H_{inj},H_{sur},H_N$ and $H_E$. Actually, hypotheses $H_{tar},H_{inj}$ and $H_{sur}$ are the targeted properties: they imply that map $\phi$ is a bijection. Hypotheses $H_N$ and $H_E$ contain the fundamental structure of the multipartite graph series. Hypothesis $H_N$ is essential, it gives a complete characterisation of the neighbourhood of a vertex on the lower levels. Hypothesis $H_E$ shows that the condition of equality of the neighbourhood at level $k-3$ of the children of $x$ in Definition~\ref{CleanFactor} actually induce a control of the neighbourhoods of the children of $x$ on all the lower levels.

\begin{description}
\item[$H_{tar}(k)$]:\label{strict} If $x\in V_k$ then $O_1(x)\subsetneq \ldots \subsetneq O_{k-1}(x)$ and $(O_1(x), \ldots , O_{k-1}(x))\in \cur{O}^{k-1}$,
\item[$H_{inj}(k)$]:\label{inj} If $x,y\in V_k$ then $x\neq y$ implies $S(x)\neq S(y)$,
\item[$H_{sur}(k)$]:\label{surj} For any sequence $(O_1,\ldots ,O_{k-1})\in \cur{O}^{k-1}$ with $O_1\subsetneq \ldots \subsetneq O_{k-1}$, there exists $x\in V_k$ such that $S(x)=(O_1,\ldots ,O_{k-1})$.
\medskip
\item[$H_N(k)$]: If $x \in V_k$, then for all $j$ such that $2\leq j<k$, $N_j(x)=W_j(x)$; where $W_j(x)$ is the set $\set{y\in V_j\ |\ (O_1(y),\ldots ,O_{j-2}(y))=(O_1(x),\ldots ,O_{j-2}(x)) \text{ and }  O_{j-1}(x)\subseteq O_{j-1}(y)\subseteq O_j(x)}$ (using the convention $(O_1(.),\ldots,O_{j-2}(.))=()$ for $j=2$).
\item[$H_E(k)$], for $k \geq 4$: If  $y_1, y_2 \in V_{k}$ such that $N_{k-2}(y_1) = N_{k-2}(y_2)$ then $\forall p\in\int{0,k-2}\setminus\set{1}, N_p(y_1)=N_p(y_2)$.
\end{description}

At each level $k$ of recursion, we start by showing $H_N(k)$, using $H_E(k-1$), then we use it to prove $H_{arr}(k),H_{inj}(k),$ and $H_{sur}(k)$, and we finish by proving $H_E(k)$.

\medskip
%
%

\medskip
\noindent\textbf{Initialisation step.} We will prove that $H_{N}(3)$, $H_{tar}(3)$, $H_{inj}(3)$, and $H_{sur}(3)$ are true. We do not prove $H_E(3)$ since it is undefined. It is worth to note that we do not need $H_E(3)$ in the proof of $H_N(4)$: instead, we use Definition~\ref{CleanFactor} that provide us the initialisation we need.

\noindent\textbf{Proof of $H_{N}(3)$.}
Since in $H_{N}(k)$, $j$ varies from $2$ to $k-1$, then in $H_{N}(3)$, $j$ only takes the value $2$. Then, to prove $H_{N}(3)$, we just have to prove that for all $x\in V_3$, $N_2(x)$ is equal to $W_2(x)=\set{y\in V_2\ |\ O_1(x)\subseteq O_1(y)\subseteq O_2(x)}$.
Let $x\in V_3$; we first show that $N_2(x)\subseteq W_2(x)$. We denote by $a_1,\ldots ,a_l$, with $l\geq 2$, the elements of the set $N_2(x)$.
Clearly, for any $i\in\int{1,l}$, $\bigcap_{1\leq i\leq l} N_0(a_i)\subseteq N_0(a_i)\subseteq \bigcup_{1\leq i\leq l} N_0(a_i)$. From Definition~\ref{DefCharSeq}, we have $N_0(a_i)=O_1(a_i)$, $O_1(x)=N_0(x)=\bigcap_{1\leq i\leq l} N_0(a_i)$, and $K(O_2(x))=\bigcap_{1\leq i\leq l} N_1(a_i)$. Moreover, from Remark~\ref{V1}, $N_1(a_i)=K(O_1(a_i))$, so we have $K(O_2(x))=\bigcap_{1\leq i\leq l} K(O_1(a_i))$.
And since, by definition, $O_2(x)\in\cur{O}'$, then, from Remark~\ref{PteOK}, $\bigcup_{1\leq i\leq l} O_1(a_i)\subseteq O_2(x)$. Consequently, for any $i\in\int{1,l}$, $O_1(x)\subseteq O_1(a_i)\subseteq O_2(x)$. That is $N_2(x)\subseteq W_2(x)$.

Conversely, we show that if $y\in W_2(x)$, then $y \in N_2(x)$. To that purpose, we show that $N_0(x)\subseteq N_0(y)$ and $N_1(x)\subseteq N_1(y)$, which implies, by maximality of $x$ in $V^*_3$ (see Definition~\ref{def:factorisation}), that $y\in N_2(x)$. First, we have $N_0(x)=O_1(x)$, and since $y\in W_2(x)$, we also have $O_1(x)\subseteq O_1(y)=N_0(y)$. Then, $N_0(x)\subseteq N_0(y)$. Since $O_1(y)\subseteq O_2(x)$, we have $K(O_2(x))\subseteq K(O_1(y))$. And from Remark~\ref{V1}, we have $K(O_1(y))=N_1(y)$ and $K(O_2(x))=N_1(x)$. Thus, $N_1(x)\subseteq N_1(y)$ and we conclude that $y\in N_2(x)$. Finally, we showed that $N_2(x)=W_2(x)$, and so $H_{N}(3)$ is true.\\

\noindent\textbf{Proof of $H_{tar}(3)$.}
Since, from $H_{N}(3)$, $N_2(x)=W_2(x)$ and since $|N_2(x)|\geq 2$, it follows that $O_1(x)\subsetneq O_2(x)$, otherwise $W_2(x)$ would contain at most one element. By definition of $V^*_3$, $|N_0(x)|\geq 2$. Since $O_1(x)=|N_0(x)|$, it follows that $O_1(x)$, and so $O_2(x)$, contains at least two elements. Moreover, from Remark~\ref{V1}, we have $K(O_2(x))=N_1(x)$. And from Definition~\ref{CleanFactor}, $|N_1(x)|\geq 2$. It follows that $K(O_2(x))$ contains at least two elements, and so does $K(O_1(x))$ since $K(O_2(x))\subseteq K(O_1(x))$.
Thus $O_1(x)$ and $O_2(x)$ both belong to $\cur{O}$: $H_{tar}(3)$ is true.

\noindent\textbf{Proof of $H_{inj}(3)$.}
For any $z\in V_3$, $N_2(z)=W_2(z)$. Thus, for any $x,y\in V_3$, $(O_1(x), O_2(x))=(O_1(y),O_2(y))$ implies that $N_2(x)=N_2(y)$, which implies that $x=y$. So $H_{inj}(3)$ holds.\\

\noindent\textbf{Proof of $H_{sur}(3)$.}
Let $O_1,O_2\in \cur{O}$ such that $O_1\subsetneq O_2$. We will find an element $x$ of $V_3$ such that $O_1(x) = O_1$ and $O_2(x) = O_2$.
Let $Y=\set{y\in V_{2}\ |\  O_{1}\subseteq O_1(y) \subseteq O_{2}}$. Let $x=Y \cup\bigcap_{y\in Y} N(y)$, we prove that $x$ is the desired element.

First, to prove that $x \in V_3$, we must prove that $x\in V^*_3$, that is $|x \cap V_2| \geq 2$ and $|x \cap V_1| \geq 2$ and $|x \cap V_0| \geq 2$.
From Lemma~\ref{Levels012} , there exists two distinct elements $y_1,y_2\in V_2$ such that $O_1(y_1)=O_1$ and $O_1(y_2)=O_2$. Clearly, $\set{y_1,y_2}\subseteq Y$, which gives $|x \cap V_2| \geq 2$.
Furthermore, $x\cap V_1=\bigcap_{y\in Y} N_1(y)=\bigcap_{y\in Y} K(O_1(y))$, from Remark~\ref{V1}. Since, for any $y\in Y$, $O_1(y)\subseteq O_2$, we have also $K(O_2)\subseteq K(O_1(y))$. It follows that $K(O_2)\subseteq x\cap V_1$, and $O_2\in\cur{O}$ implies that $|K(O_2)|\geq 2$, then $|x\cap V_1|\geq 2$.
We also have $x\cap V_0=\bigcap_{y\in Y} N_0(y)=\bigcap_{y\in Y} O_1(y)$. And since $O_1\subseteq O_1(y)$ for all $y\in Y$, then $O_1\subseteq x\cap V_0$. It follows that $|x \cap V_0| \geq 2$.

We now show that $x$ is maximal in $V^*_3$. Let $z\in V_{2} \setminus Y$ \textit{i.e.} $O_1\not\subseteq O_1(z)$ or $O_1(z)\not\subseteq O_{2}$, we prove that $x\cap V_0\not\subseteq N_0(z)$ or $x\cap V_1\not\subseteq N_1(z)$, which implies that $\bigcap_{y\in Y\cup\set{z}} N(y)\subsetneq \bigcap_{y\in Y} N(y)$ and that $x$ is maximal, since this holds for all $z\in V_{2} \setminus Y$.
Let us first note that $x\cap V_0= \bigcap_{y \in Y} N_0(y) = \bigcap_{y \in Y} O_1(y) = O_1$, the last equality coming from the fact that $\forall O\in\cur{O}, \exists z\in V_2, O_1(z)=O$  (see Lemma~\ref{Levels012}).
On the other hand $x\cap V_1 = \bigcap_{y \in Y} N_1(y) = \bigcap_{y \in Y} K(O_1(y)) = K(\bigcup_{y \in Y} O_1(y)) = K(O_2)$, again using Lemma~\ref{Levels012} for the last equality.
Now, if $O_1(z)\not\subseteq O_{2}$, since $O_1(z)\in\cur{O}'$, then, from Remark~\ref{Kincl}, $K(O_2) \not\subseteq K(O_1(z))$. And since, from what precedes, $K(O_2)=x\cap V_1$ and, from Remark~\ref{V1}, $K(O_1(z))=N_1(z)$, we obtain $x\cap V_1\not\subseteq N_1(z)$.
On the other hand, if $O_1\not\subseteq O_1(z)$, since $O_1=x\cap V_0$ (see above) and $O_1(z)=N_0(z)$ (by definition), we obtain $x\cap V_0\not\subseteq N_0(z)$.
Consequently, $x$ is maximal in $V^*_3$ and then $x\in V_3$.

The last condition we have to check for proving $H_{sur}(3)$ is that $O_1(x)= O_1$ and $O_2(x)=O_2$. We already have $O_1=N_0(x)=O_1(x)$. Moreover by definition, $O_2(x)$ is the unique element of $\cur{O}'$ such that $K(O_2(x)) = \bigcap_{y \in Y} N_1(y)$, and we also have from above that $K(O_2) = \bigcap_{y \in Y} N_1(y)$, then $O_2(x)= O_2$.

Finally, $H_{sur}(3)$ is true.


\medskip
\noindent\textbf{Recursion step.}
Now, let us suppose that $k\geq 4$ and that for all $i$ such that $3 \leq i <k$, $H_{tar}(i)$, $H_{inj}(i)$, $H_{sur}(i)$ and $H_{E}(i)$ are true. Note that we did not prove $H_E(3)$, which is not even defined, but we don't need it. Actually, in step $k$ of the recursion, $H_E(k-1)$ is used only in the proof of $H_N(k)$. For proving $H_N(4)$, the use of $H_E(3)$ is replaced by the use of the definition of $V_3^*$ (Definition~\ref{CleanFactor}).\\

\noindent\textbf{Proof of $H_{N}(k)$.}
Let $x\in V_k$. We denote by $a_1,\ldots ,a_l$ the elements of the set $N_{k-1}(x)$. 
Let $i_1,i_2\in\int{1,l}$. 
If $k \geq 5$, by hypothesis $H_E(k-1)$, we have that $N_p(a_{i_1})=N_p(a_{i_2})=N_p(x)$ for all  $p\in\int{0,k-3}\setminus\set{1}$. If $k=4$, by Definition~\ref{CleanFactor}, we have that $N_0(a_{i_1})=N_0(a_{i_2})=N_0(x)$. Thus, independently of the value of $k\geq 4$ we have $N_p(a_{i_1})=N_p(a_{i_2})=N_p(x)$ for all $p\in\int{0,k-3}\setminus\set{1}$.
Then, using the Definition~\ref{DefCharSeq} of the characterising sequence, it follows that, for $p=0$, we obtain $O_1(a_i)=O_1(a_j)$ and for $p\in\int{2,k-3}$, we obtain $O_p(a_i)=O_p(a_j)= O_p(x)$. 

Let $j\in\int{2,k-3}$ and let $i\in\int{1,l}$. From recursion hypothesis $H_{N}(k-1)$ applied to $a_i$, we have $N_j(a_i)=W_j(a_i)$. Since for all $q\in \int{1,k-3}$ we have $O_q(x)=O_q(a_i)$, then, from the definition of $W_j$, it follows that $W_j(x)=W_j(a_i)$. Finally, since $N_j(x)=N_j(a_i)$, we obtain $N_j(x)=W_j(x)$, for all $2\leq j\leq k-3$. Then we just have to prove that $N_{k-2}(x)=W_{k-2}(x)$ and $N_{k-1}(x)=W_{k-1}(x)$.\\

We start with $N_{k-1}(x)$. We will first show that for any $i\in\int{1,l}$ we have $a_i\in W_{k-1}(x)$, that is $N_{k-1}(x)\subseteq W_{k-1}(x)$. As explained above, we already know that, for all $q\in \int{1,k-3}$, $O_q(x)=O_q(a_i)$. Then, we only need to show that $O_{k-2}(x)\subseteq O_{k-2}(a_i)\subseteq O_{k-1}(x)$.

Let us show the first inclusion: $O_{k-2}(x)\subseteq O_{k-2}(a_i)$. By definition, $N_{k-2}(x) = \bigcap_{t\in\int{1,l}}N_{k-2}(a_t)\subseteq N_{k-2}(a_i)$. Then, we have $\bigcap_{b\in N_{k-2}(a_i)} N_1(b)\subseteq \bigcap_{b\in N_{k-2}(x)} N_1(b)$. And since, by definition, $\bigcap_{b\in N_{k-2}(a_i)} N_1(b)=K(O_{k-2}(a_i)$ and $\bigcap_{b\in N_{k-2}(x)} N_1(b)=K(O_{k-2}(x)$, then we obtain $K(O_{k-2}(a_i))\subseteq K(O_{k-2}(x))$. Thus, from Remark~\ref{Kincl}, we have $O_{k-2}(x)\subseteq O_{k-2}(a_i)$.


Let us now show the second inclusion: $O_{k-2}(a_i)\subseteq O_{k-1}(x)$. Since $N_{k-2}(a_i)\subseteq \bigcup_{t\in\int{1,l}}N_{k-2}(a_t)$ then $\bigcap_{b\in N_{k-2}(a_i)}N_1(b)\supseteq \bigcap_{b\in\bigcup_{t\in\int{1,l}}N_{k-2}(a_t)}N_1(b)$, and by definition $\bigcap_{b\in N_{k-2}(a_i)}N_1(b)=K(O_{k-2}(a_i))$. In order to show the inclusion we aim at, we will show that $\bigcap_{b\in\bigcup_{t\in\int{1,l}}N_{k-2}(a_t)}N_1(b)=\bigcap_{t\in\int{1,l}}N_1(a_t)$. Since, by definition again, $\bigcap_{t\in\int{1,l}}N_1(a_t)=K(O_1(x))$, this will give $K(O_{k-2}(a_i))\supseteq K(O_{k-2}(x))$, which implies $O_{k-2}(a_i))\subseteq O_{k-2}(x)$, the inclusion we aim at.

Then, let us show the equality $\bigcap_{b\in\bigcup_{t\in\int{1,l}}N_{k-2}(a_t)}N_1(b)=\bigcap_{t\in\int{1,l}}N_1(a_t)$ by a double inclusion. Let $z\in\bigcap_{t\in\int{1,l}}N_1(a_t)$, then for all $t\in\int{1,l}$, $z\in N_1(a_t)$. It follows that for all $b\in N_{k-2}(a_t)$, $z\in N_1(b)$. Since this holds for all $t\in\int{1,l}$ and for all $b\in N_{k-2}(a_t)$, then we obtain $z\in\bigcap_{b\in\bigcup_{t\in\int{1,l}}N_{k-2}(t)}N_1(b)$. Conversely, let $z\in\bigcap_{b\in\bigcup_{t\in\int{1,l}}N_{k-2}(a_t)}N_1(b)$, we show that $z\in\bigcap_{t\in\int{1,l}}N_1(a_t)$. For all $b\in\bigcup_{t\in\int{1,l}}N_{k-2}(a_t)$ we have $z\in N_1(b)$. In particular, for any $t\in\int{1,l}$ and for any $b\in N_{k-2}(a_t)$, we have $z\in N_1(b)$, and so $z\in N_1(a_t)$. As this holds for any $t\in\int{1,l}$, then $z\in\bigcap_{t\in\int{1,l}}N_1(a_t)$, that is $\bigcap_{b\in\bigcup_{t\in\int{1,l}}N_{k-2}(t)}N_1(b)\subseteq\bigcap_{t\in\int{1,l}}N_1(a_t)$. Finally, as we already showed the converse inclusion, we obtain $\bigcap_{b\in\bigcup_{t\in\int{1,l}}N_{k-2}(t)}N_1(b)=\bigcap_{t\in\int{1,l}}N_1(a_t)$.

We can now finish the proof of the inclusion $N_{k-1}(x)\subseteq W_{k-1}(x)$. Remember that $\bigcap_{b\in N_{k-2}(a_i)}N_1(b)\supseteq \bigcap_{b\in\bigcup_{t\in\int{1,l}}N_{k-2}(t)}N_1(b)$ and that by definition $\bigcap_{b\in N_{k-2}(a_i)}N_1(b)=K(O_{k-2}(a_i)$.\\ In addition, we just proved that $\bigcap_{b\in\bigcup_{t\in\int{1,l}}N_{k-2}(t)}N_1(b)=\bigcap_{t\in\int{1,l}}N_1(a_t)$, and, by definition again, we have $\bigcap_{t\in\int{1,l}}N_1(a_t)=K(O_{k-1}(x))$. We then obtain $K(O_{k-2}(a_i))$ $\supseteq K(O_{k-1}(x))$, and Remark~\ref{Kincl} concludes that $O_{k-2}(a_i)\subseteq O_{k-1}(x)$, for all $i\in\int{1,l}$. So finally, putting everything together (we proved above that $O_{k-2}(x)\subseteq O_{k-2}(a_i)$) we get $O_{k-2}(x)\subseteq O_{k-2}(a_i)\subseteq O_{k-1}(x)$, for all $i\in\int{1,l}$, which completes our proof of $a_i\in W_{k-1}(x)$, for all $i\in\int{1,l}$, that is $N_{k-1}(x)\subseteq W_{k-1}(x)$.\\

Let us now prove the converse inclusion: $W_{k-1}(x)\subseteq N_{k-1}(x)$.
Let $y\in W_{k-1}(x)$, we have $O_q(y)=O_q(x)$ for all $q\in\int{1,k-3}$. Moreover, as we showed at the beginning of the proof of $H_{N}(k)$, for all $i\in\int{1,l}$, and all $q\in\int{1,k-3}$ we have $O_q(a_i)=O_q(x)$ and so $O_q(a_i)=O_q(y)$. For $q=1$, from the definition of $O_1$, this gives $N_0(a_i)=N_0(y)$. For $q\geq 2$ (which occurs only for $k\geq 5$), using recursion hypothesis $H_N(k-1)$ that characterises, for any $z\in V_{k-1}$ and any $q\geq 2$, $N_q(z)$ as a function of only $O_1(z),\ldots,O_q(z)$, we then obtain that since $y$ and $a_i$ have the same sequence $O_1(y),\ldots,O_{k-3}(y)=O_1(a_i),\ldots,O_{k-3}(a_i)$, they necessarily have the same neighbourhood $N_q(y)=N_q(a_i)$, for all $q\in\int{2,k-3}$. Since we showed that $y$ and $a_i$ also have the same neighbourhood on $V_0$, and since we showed at the beginning of the proof of $H_N(k)$ that $N_p(a_i)=N_p(x)$ for all $p\in\int{0,k-3}\setminus\set{1}$, it follows that for all $p\in\int{0,k-3}\setminus\set{1}$, we have $N_p(y)=N_p(x)$.
Then, in order to show that $y\in N_{k-1}(x)$, we only need to show that $N_1(x)\subseteq N_1(y)$ and $N_{k-2}(x)\subseteq N_{k-2}(y)$, which implies, by maximality of $x$ (see Definition~\ref{def:factorisation}), that $y\in N_{k-1}(x)$. First, let us show that $N_1(x)\subseteq N_1(y)$. Since $y\in W_{k-1}(x)$, we have $O_{k-2}(y)\subseteq O_{k-1}(x)$, which implies $K(O_{k-1}(x))\subseteq K(O_{k-2}(y))$. And since, from remark~\ref{V1}, $K(O_{k-1}(x))=N_1(x)$ and $K(O_{k-2}(y))=N_1(y)$, then we get the desired inclusion: $N_1(x)\subseteq N_1(y)$.

Let us now show that $N_{k-2}(x)\subseteq N_{k-2}(y)$. Let $z\in N_{k-2}(x)$. By recursion hypothesis $H_{N}(k-1)$, we know that $N_{k-2}(y)=W_{k-2}(y)$. Thus, our aim is to show that $z\in W_{k-2}(y)$, that is $O_q(z)=O_q(y)$ for all $q\in\int{1,k-4}$ and $O_{k-3}(y)\subseteq O_{k-3}(z)\subseteq O_{k-2}(y)$. Let $i\in\int{1,l}$, since $z\in N_{k-2}(x)$ then $z\in N_{k-2}(a_i)$. It follows by recursion hypothesis $H_{N}(k-1)$ that for all $q\in\int{1,k-4}, O_q(z)=O_q(a_i)$. And since we already showed that $O_q(a_i)=O_q(x)$, then we obtain $O_q(z)=O_q(x)$ for all $q\in\int{1,k-4}$. On the other hand, since $y\in W_{k-1}(x)$, then $O_q(y)=O_q(x)$ for all $q\in\int{1,k-4}$, which finally gives that $O_q(y)=O_q(z)$ for all $q\in\int{1,k-4}$.

Let us now show that $O_{k-3}(z)\subseteq O_{k-2}(y)$. Since $y\in W_{k-1}(x)$, we have $O_{k-2}(x)\subseteq O_{k-2}(y)$. Then, it is sufficient to show that $O_{k-3}(z)\subseteq O_{k-2}(x)$. We state this fact as a proposition as we use it further in the proof:\\
\textbf{(Prop. A)} for any vertex $z\in N_{k-2}(x)$, we have $O_{k-3}(z)\subseteq O_{k-2}(x)$.\\
Clearly, since $z\in N_{k-2}(x)$, we have $N_1(z)\supseteq \bigcap_{b\in N_{k-2}(x)}N_1(b)$. By definition, $N_1(z)=K(O_{k-3}(z)$ and $\bigcap_{b\in N_{k-2}(x)}N_1(b)=K(O_{k-2}(x))$. We then obtain $K(O_{k-3}(z))$ $\supseteq K(O_{k-2}(x))$, which gives, from Remark~\ref{Kincl}, $O_{k-3}(z)\subseteq O_{k-2}(x)$. And since $y\in W_{k-1}(x)$, we have $O_{k-2}(x)\subseteq O_{k-2}(y)$. Thus, $O_{k-3}(z)\subseteq O_{k-2}(y)$.

We now show that $O_{k-3}(y)\subseteq O_{k-3}(z)$. Since $y\in W_{k-1}(x)$, then $O_{k-3}(y)=O_{k-3}(x)$. As we already showed, for any $i\in\int{1,l}$, we have $O_{k-3}(x)=O_{k-3}(a_i)$, and so $O_{k-3}(y)=O_{k-3}(a_i)$. Moreover, since $z\in N_{k-2}(x)$ then $z\in N_{k-2}(a_i)$. And by recursion hypothesis $H_{N}(k-1)$, we have $W_{k-2}(a_i)=N_{k-2}(a_i)$. Thus, $z\in W_{k-2}(a_i)$ satisfies $O_{k-3}(a_i)\subseteq O_{k-3}(z)$. Finally, we obtain $O_{k-3}(y)\subseteq O_{k-3}(z)$, which achieves the proof of $z\in N_{k-2}(y)$. Thus, we have $N_{k-2}(x)\subseteq N_{k-2}(y)$.

In summary, we showed that for any $y\in W_{k-1}(x)$, we have $N_p(y)=N_p(x)$ for all $p\in\int{0,k-3}\setminus\set{1}$ and $N_1(x)\subseteq N_1(y)$ and $N_{k-2}(x)\subseteq N_{k-2}(y)$. Then, by maximality of $x$ (see Definition~\ref{def:factorisation}), $y$ belongs to $N_{k-1}(x)$. That is, $W_{k-1}(x)\subseteq N_{k-1}(x)$. As we also showed the converse inclusion, we obtain $N_{k-1}(x)=W_{k-1}(x)$.

In order to achieve the proof of $H_{N}(k)$, we still have to show that $N_{k-2}(x)=W_{k-2}(x)$. We start with $W_{k-2}(x)\subseteq N_{k-2}(x)$. Let $z\in W_{k-2}(x)$. Let $i\in\int{1,l}$, we show that $z\in N_{k-2}(a_i)$. By recursion hypothesis $H_{N}(k-1)$ we know that $N_{k-2}(a_i)=W_{k-2}(a_i)=\set{w\in V_{k-2}\ |\ (O_1(w),\ldots ,O_{k-4}(w))=(O_1(a_i),\ldots ,O_{k-4}(a_i)) \text{ and }  O_{k-3}(a_i)$ $\subseteq O_{k-3}(w)\subseteq O_{k-2}(a_i)}$. As we already mentioned several times, we have $(O_1(a_i),\ldots ,$ $O_{k-4}(a_i))=(O_1(x),\ldots ,O_{k-4}(x))$. Since $z\in W_{k-2}(x)$, we also have $(O_1(z),\ldots ,$ $O_{k-4}(z))=(O_1(x),\ldots ,O_{k-4}(x))$, and then $(O_1(z),\ldots ,O_{k-4}(z))=(O_1(a_i),\ldots ,$ $O_{k-4}(a_i))$. Again, since $z\in W_{k-2}(x)$, we have $O_{k-3}(x)\subseteq O_{k-3}(z)\subseteq O_{k-2}(x)$. Since we already proved that $N_{k-1}(x)=W_{k-1}(x)$, we know that $O_{k-3}(a_i)=O_{k-3}(x)$ and $O_{k-2}(x)\subseteq O_{k-2}(a_i)$. It follows that $O_{k-3}(a_i)\subseteq O_{k-3}(z)\subseteq O_{k-2}(a_i)$, which shows that $z\in N_{k-2}(a_i)$. As this holds for any $i\in\int{1,l}$, then we conclude that $z\in N_{k-2}(x)$.

Conversely, let $z\in N_{k-2}(x)$. Then, for all $i\in\int{1,l}$, $z\in N_{k-2}(a_i)$. From recursion hypothesis $H_{N}(k-1)$ applied to $a_i$, we get $(O_1(z),\ldots ,O_{k-4}(z))=(O_1(a_i),\ldots ,$ $O_{k-4}(a_i))$. And for the same reason, we also have $O_{k-3}(a_i)\subseteq O_{k-3}(z)$. As we know that $(O_1(a_i),\ldots ,O_{k-3}(a_i))=(O_1(x),\ldots ,O_{k-3}(x))$, we obtain $(O_1(z),\ldots ,O_{k-4}(z))$ $=(O_1(x),\ldots ,O_{k-4}(x))$ and $O_{k-3}(x)\subseteq O_{k-3}(z)$. Then, the only thing left we have to show in order to prove that $z\in W_{k-2}(x)$, which is our goal, is to prove that $O_{k-3}(z)\subseteq O_{k-2}(x)$. In fact, we already proved this proposition in the proof of $W_{k-1}(x)\subseteq N_{k-1}(x)$ above, referred as (Prop. A) in the text. So we finally obtain that $N_{k-2}(x)\subseteq W_{k-2}(x)$, and since we already proved the converse inclusion, we obtain the equality between the two sets, $N_{k-2}(x)=W_{k-2}(x)$, which completes our proof of $H_{N}(k)$.\\

\noindent\textbf{Proof of $H_{tar}(k)$.}
From $H_{N}(k)$ we know that for any $i\in\int{1,l}$, $(O_1(a_i),\ldots , O_{k-3}(a_i))$ $=(O_1(x),\ldots , O_{k-3}(x))$. Then, from recursion hypothesis $H_{tar}(k-1)$, we have $O_1(x)\subsetneq \ldots \subsetneq O_{k-3}(x)$ and, in particular, we have $O_{1}(x)\in\cur{O}$ and so $|O_1(x)|\geq 2$. Since $N_{k-2}(x)=W_{k-2}(x)$ and $|N_{k-2}(x)|>1$, necessarily $O_{k-3}(x)\subsetneq O_{k-2}(x)$.
Similarly, the fact that $N_{k-1}(x)=W_{k-1}(x)$ and $|N_{k-1}(x)|>1$ implies that $O_{k-2}(x)\subsetneq O_{k-1}(x)$. At last, from Remark~\ref{V1}, we have $K(O_{k-1}(x))=N_1(x)$, and since $|N_1(x)|\geq 2$, it follows that $|K(O_{k-1}(x))|\geq 2$. Combined with the fact that $|O_1(x)|\geq 2$, this implies that for all $j\in\int{1,k-1}$, we have $O_j(x)\in\cur{O}$. Thus, $H_{tar}(k)$ is true.\\

\noindent\textbf{Proof of $H_{inj}(k)$.}
Let $x,x'\in V_k$ such that $S(x)=S(x')$. From $H_{N}(k)$, $N_{k-1}(x)=W_{k-1}(x)$ and $N_{k-1}(x')=W_{k-1}(x')$. And since $S(x)=S(x')$, we have $W_{k-1}(x)=W_{k-1}(x')$. As a consequence, $N_{k-1}(x)=N_{k-1}(x')$ and so $x=x'$. Therefore $H_{inj}(k)$ is true.\\

\noindent\textbf{Proof of $H_{sur}(k)$.}
Let $(O_1,\ldots ,O_{k-1})\in \cur{O}^{k-1}$ such that $O_1\subsetneq \ldots \subsetneq O_{k-1}$.
From recursion hypothesis $H_{sur}(k-1)$, for any $P\in\cur{O}$ such that $O_{k-3}\subsetneq P$, there exists $y_P\in V_{k-1}$ such that $S(y_P)=(O_1,\ldots , O_{k-3}, P)$. We denote by $Y$ the set $Y=\set{y\in V_{k-1}\ |\ (O_1(y),\ldots , O_{k-3}(y))=(O_1,\ldots , O_{k-3}) \text{ and } O_{k-2}\subseteq O_{k-2}(y)\subseteq O_{k-1}}$.
Let $x=Y\cup\bigcap_{y\in Y} N(y)$. We will show that $x$ is maximal in $V_k^*$ and that the corresponding element of $V_k$ has the desired sequence $S(x)=(O_1,\ldots,O_{k-1})$.

Let us start by showing that $x\in V_k^*$. Since $O_{k-2}\subsetneq O_{k-1}$, we have $|Y|\geq 2$, that is $|x\cap V_{k-1}|\geq 2$.
From recursion hypothesis $H_{N}(k-1)$, for any $y\in Y$, $N_{k-2}(y)=\set{t\in V_{k-2}\ |\ (O_1(t),\ldots,O_{k-4}(t))=(O_1,\ldots,O_{k-4})$ and $O_{k-3}\subseteq O_{k-3}(t)\subseteq O_{k-2}(y)}$. And since, by definition, $O_{k-2}\subseteq O_{k-2}(y)$, then $N_{k-2}(y)$ contains at least the two elements of $V_{k-2}$ having characterising sequences $(O_1,\ldots ,O_{k-4},O_{k-3})$ and $(O_1,\ldots ,O_{k-4},O_{k-2})$, which do exist from recursion hypothesis $H_{sur}$. Since this is true for all $y\in N_{k-2}(x)$, then $x$ itself has these two elements as neighbours on level $V_{k-2}$. Then, $|x\cap V_{k-2}|\geq 2$. Let us now show that $|x\cap V_1|\geq2$. For any $y\in Y$, from Remark~\ref{V1}, we have $N_1(y)=K(O_{k-2}(y))$. Since $O_{k-2}(y)\subseteq O_{k-1}$ then $K(O_{k-2}(y))\supseteq K(O_{k-1})$. Thus, we obtain $x\cap V_1=\bigcap_{y\in Y} N_1(y)=\bigcap_{y\in Y} K(O_{k-2}(y))\supseteq K(O_{k-1})$. And since $O_{k-1}\in \cur{O}$, $O_{k-1}$ contains at least two elements and so does $x\cap V_1$. In order to complete the proof of $x\in V_k^*$, we need to show that for all $y,y'\in Y$, we have $N_{k-3}(y)=N_{k-3}(y')$. First, note that, by definition, $(O_1(y),\ldots,O_{k-3}(y))=(O_1(y'),\ldots,O_{k-3}(y'))=(O_1,\ldots,O_{k-3})$. Moreover, recursion hypothesis $H_{N}(k-1)$ gives that the neighbourhood at level $V_{k-3}$ of any vertex $z\in V_{k-1}$ only depends on the piece of sequence $(O_1(z),\ldots ,O_{k-3}(z))$. And since $y$ and $y'$ have the same such pieces of sequence, it follows that $N_{k-3}(y)=N_{k-3}(y')$. Thus $x\in V^*_k$.

We will now show that $x$ is maximal in $V^*_k$. Let $z\in V_{k-1}\setminus Y$, we show that if for some $j\in\int{1,k-3}, O_j(z)\neq O_j$ then there is no element $B\in V_k^*$ containing $Y\cup\set{z}$. We denote $y\in Y$ an arbitrary element of $Y$ and we distinguish between the case where $k=4$ and the case where $k\geq 5$.

Let us start with the general case where $k\geq 5$, we show that $N_{k-3}(z)\neq N_{k-3}(y)$, which implies, from Definition~\ref{CleanFactor}, that there is no element of $V_k^*$ containing both $y$ and $z$. So let $j\in\int{1,k-3}$, such that $O_j(z)\neq O_j$. Since $O_j(y)=O_j$, then we have $O_j(z)\neq O_j(y)$. Moreover, from recursion hypothesis $H_N(k-1)$, we have $N_{k-3}(z)=W_{k-3}(z)$ and $N_{k-3}(y)=W_{k-3}(y)$. We again distinguish several cases depending on the value of $j$.\\
If $j\leq k-5$ (which may occur only when $k\geq 6$), from recursion hypothesis $H_{sur}(k-3)$, there exists $t_1\in V_{k-3}$ such that $S(t_1)=(O_1(z),\ldots ,O_{k-4}(z))$. Clearly, from the definition of  $W_j(z)$, we have $t_1\in W_{k-3}(z)=N_{k-3}(z)$. On the opposite, from the definition of $W_j(y)$, and since $O_j(z)\neq O_j(y)$ with $j\leq k-5$, we obtain $t_1\not\in W_{k-3}(y)=N_{k-3}(y)$ and it follows that $N_{k-3}(z)\neq N_{k-3}(y)$.\\
If $j=k-4$. Since $O_{k-4}(z)\neq O_{k-4}(y)$, then one of the two sets $O_{k-4}(z), O_{k-4}(y)$ is not included in the other, say $O_{k-4}(y)\not\subseteq O_{k-4}(z)$ without loss of generality. Consider again the element $t_1\in N_{k-3}(z)$ described above. Since $O_{k-4}(t_1)=O_{k-4}(z)\not\supseteq O_{k-4}(y)$, then it follows, from the definition of $W_{k-4}(y)$, that $t_1\not\in W_{k-3}(y)=N_{k-3}(y)$, and so $N_{k-3}(z)\neq N_{k-3}(y)$.\\
If $j=k-3$, without loss of generality we can assume that $O_{k-3}(z)\not\subseteq O_{k-3}(y)$. From recursion hypothesis $H_{sur}(k-3)$, there exists $t_2\in V_{k-3}$ such that $O_{k-4}(t_2)=O_{k-3}(z)$ and $(O_1(z),\ldots ,O_{k-5}(z))=(O_1(z),\ldots ,O_{k-5}(z))$ (using, as usual, the convention $(O_1(z),\ldots ,O_{k-5}(z))=()$ if $k=5$). From the definition of $W_{k-3}(z)$ and $W_{k-3}(y)$, we obtain that $t_2\in W_{k-3}(z)=N_{k-3}(z)$ but, since $O_{k-3}(z)\not\subseteq O_{k-3}(y)$, we have $t_2\not\in W_{k-3}(y)=N_{k-3}(y)$. Thus, in all cases where $k\geq 5$, if there exists some $j\in\int{1,k-3}$ such that $O_j(z)\neq O_j$, then $N_{k-3}(z)\neq N_{k-3}(y)$.

Let us now deal with the particular case where $k=4$. In this case, necessarily the index $j\in\int{1,k-3}$ such that $O_j(z)\neq O_j$ is $j=1$. We immediately obtain that $N_0(z)=O_1(z)\neq O_1(y)=N_0(y)$. Then, from Definition~\ref{CleanFactor} (case $k=4$), there is no element $B\in V_k^*$ containing both $y$ and $z$.\\
Finally, we conclude that, regardless of the value of $k\geq 4$, if for some $j\in\int{1,k-3}, O_j(z)\neq O_j$ then there is no element $B\in V_k^*$ containing $Y\cup\set{z}$.

Thus, to show that $x$ is maximal in $V_k^*$ we only need to show that for any $z\in V_{k-1}$ such that $(O_1(z),\ldots , O_{k-3}(z))=(O_1,\ldots , O_{k-3})$, if $O_{k-2}\not\subseteq O_{k-2}(z)$ or $O_{k-2}(z)\not\subseteq O_{k-1}$, then we have $\bigcap_{y\in Y} N(y)\not\subseteq N(z)$.

We first treat the case where $O_{k-2}(z)\not\subseteq O_{k-1}$. In this case, from Remark~\ref{Kincl}, $K(O_{k-1})\not\subseteq K(O_{k-2}(z))$. From Remark~\ref{V1}, we have $K(O_{k-2}(z))=N_1(z)$. We now show that $K(O_{k-1})=\bigcap_{y\in X\cap V_{k-1}} N_1(y)$, which will give us the desired result: $\bigcap_{y\in x\cap V_{k-1}} N_1(y)\not\subseteq N_1(z)$.
From Remark~\ref{V1}, for any $y\in V_{k-1}$, $N_1(y)=K(O_{k-2}(y))$. It follows that $\bigcap_{y\in x\cap V_{k-1}} N_1(y)=\bigcap_{y\in x\cap V_{k-1}} K(O_{k-2}(y))$ and we also have $\bigcap_{y\in x\cap V_{k-1}} K(O_{k-2}(y))=\bigcap_{P\in\cur{O} \text{ and } O_{k-2}\subseteq P\subseteq O_{k-1}} K(P)$, from $H_{sur}(k-1)$ and the definition of $x$. From Lemma~\ref{PteOK}, we get $\bigcap_{P\in\cur{O} \text{ and } O_{k-2}\subseteq P\subseteq O_{k-1}} K(P)=$ $K(\bigcup_{P\in\cur{O} \text{ and } O_{k-2}\subseteq P\subseteq O_{k-1}} P)$, which is clearly equal to $K(O_{k-1})$. And so we have $\bigcap_{y\in x\cap V_{k-1}} N_1(y)=K(O_{k-1})$. As a consequence, we obtain $\bigcap_{y\in x\cap V_{k-1}} N_1(y)\not\subseteq N_1(z)$. Then, adding $z$ to $x\cap V_{k-1}$ would strictly decrease $\bigcap_{y\in x\cap V_{k-1}} N_1(y)$.

Let us now consider the case where $O_{k-2}\not\subseteq O_{k-2}(z)$.
Using recursion hypothesis $H_{N}(k-1)$, for all $y\in x\cap V_{k-1}$ we have $N_{k-2}(y)=\set{t\in V_{k-2}\ |\ O_{k-3}(y) \subseteq O_{k-3}(t) \subseteq O_{k-2}(y) \text{ and } (O_1(t),\ldots ,O_{k-4}(t))=(O_1,\ldots ,O_{k-4})}$. Let us denote $Z=\set{t\in V_{k-2}\ |\ \bigcup_{y\in x\cap V_{k-1}} O_{k-3}(y) \subseteq O_{k-3}(t) \subseteq \bigcap_{y\in x\cap V_{k-1}} O_{k-2}(y) \text{ and } (O_1(t),\ldots ,$ $O_{k-4}(t))=(O_1,\ldots ,O_{k-4})}$ (as usual, we use the convention $(O_1,\ldots , O_{k-4})=()$ when $k=4$). We show that $\bigcap_{y\in x\cap V_{k-1}} N_{k-2}(y)=Z$. Let $t\in \bigcap_{y\in x\cap V_{k-1}} N_{k-2}(y)$, then for all $y\in x\cap V_{k-1}$, $O_{k-3}(y) \subseteq O_{k-3}(t) \subseteq O_{k-2}(y)$ and so $\bigcup_{y\in x\cap V_{k-1}} O_{k-3}(y) \subseteq O_{k-3}(t) \subseteq \bigcap_{y\in x\cap V_{k-1}} O_{k-2}(y)$, that is $t\in Z$. Conversely, if $t\in Z$ then we have $\bigcup_{y\in x\cap V_{k-1}} O_{k-3}(y) \subseteq O_{k-3}(t) \subseteq \bigcap_{y\in x\cap V_{k-1}} O_{k-2}(y)$ and so $O_{k-3}(y) \subseteq O_{k-3}(t) \subseteq O_{k-2}(y)$ for all $y\in x\cap V_{k-1}$, that is $t\in\bigcap_{y\in x\cap V_{k-1}} N_{k-2}(y)$. Thus, $\bigcap_{y\in x\cap V_{k-1}} N_{k-2}(y)$ $=Z$. By definition, for all $y\in x\cap V_{k-1}$, $O_{k-3}(y)=O_{k-3}$ and $O_{k-2}\subseteq O_{k-2}(y)$. It follows that $\bigcup_{y\in x\cap V_{k-1}} O_{k-3}(y)=O_{k-3}$ and $O_{k-2}\subseteq \bigcap_{y\in x\cap V_{k-1}} O_{k-2}(y)$. 
Moreover, from recursion hypothesis $H_{sur}$, there exists $t'\in V_{k-2}$ such that $O_{k-3}(t')=O_{k-2}$ and $(O_1(t'),\ldots ,O_{k-4}(t'))=(O_1,\ldots ,O_{k-4})$. From what precedes, since $O_{k-3}\subseteq O_{k-2}\subseteq \bigcap_{y\in x\cap V_{k-1}} O_{k-2}(y)$, then $t'\in Z$. On the other hand, from recursion hypothesis $H_{N}(k-1)$, we have $N_{k-2}(z)=\set{t\in V_{k-2}\ |\ O_{k-3} \subseteq O_{k-3}(t) \subseteq O_{k-2}(z) \text{ and } (O_1(t),\ldots ,$ $O_{k-4}(t))=(O_1,\ldots ,O_{k-4})}$. And since $O_{k-2}\not\subseteq O_{k-2}(z)$, it follows that $t'\not\in N_{k-2}(z)$, while $t'\in Z=\bigcap_{y\in x\cap V_{k-1}} N_{k-2}(y)$. Thus, $\bigcap_{y\in x\cap V_{k-1}} N_{k-2}(y)\not\subseteq N_{k-2}(z)$ and adding $z$ to $x\cap V_{k-1}$ would strictly decrease $\bigcap_{y\in x\cap V_{k-1}} N_{k-2}(y)$. Finally, $x$ is maximal in $V^*_k$ and is therefore an element of $V_k$.

In order to conclude the proof of $H_{sur}(k)$, let us now show that the element $x$ of $V_k$ has the desired characterising sequence $(O_1,\ldots ,O_{k-1})$. First, from $H_{N}(k)$, which we already proved, we know that $(O_1(x),\ldots,O_{k-3}(x))=(O_1(y),\ldots,O_{k-3}(y))$ for any $y\in N_{k-1}(x)$, which gives $(O_1(x),\ldots,O_{k-3}(x))=(O_1,\ldots,O_{k-3})$, from the definition of $x=Y\bigcup_{y\in Y} N(y)$. Second, from Remark~\ref{V1}, we have $K(O_{k-1}(x))=N_1(x)$ and we already know that $|N_1(x)|\geq 2$ (see beginning of the proof of $H_{sur}(k)$). This gives $|K(O_{k-1}(x))|\geq 2$, and as we have $|O_1(x)|\geq 2$ and $O_1(x)\subseteq \ldots\subseteq O_{k-2}(x)\subseteq O_{k-1}(x)$, we obtain that $O_{k-2}(x),O_{k-1}(x)\in \cur{O}$. Now, from $H_{N}(k)$, we know that the couple $(O_{k-2}(x),O_{k-1}(x))$ is such that $N_{k-1}(x)=\set{y\in V_{k-1}\ |\ O_{k-2}(x) \subseteq O_{k-2}(y) \subseteq O_{k-1}(x) \text{ and } (O_1(y),\ldots ,O_{k-3}(y))=(O_1,\ldots ,O_{k-3})}$. And by definition of $x$, the couple $(O_{k-2},O_{k-1})$ also satisfies this condition. But since $O_{k-2},O_{k-1},O_{k-2}(x),$ $O_{k-1}(x)$ all belong to $\cur{O}$, then, from recursion hypothesis $H_{sur}(k-1)$, for any $P\in\set{O_{k-2},O_{k-1},O_{k-2}(x),O_{k-1}(x)}$ there exists $y\in V_{k-1}$ such that $S(y)=(O_1,\ldots,$ $O_{k-3},P)$. Then, for $P=O_{k-1}$, using the definition of $N_{k-1}(x)$ based on the couple $(O_{k-2},O_{k-1})$, we obtain that $y\in N_{k-1}(x)$, and consequently, using the definition of $N_{k-1}(x)$ based on the couple $(O_{k-2}(x),O_{k-1}(x))$, we have $O_{k-1}\subseteq O_{k-1}(x)$. Similarly, for $P=O_{k-1}(x)$ we obtain $O_{k-1}(x)\subseteq O_{k-1}$, and it follows that $O_{k-1}(x)=O_{k-1}$. Analogously, choosing $P=O_{k-2}$ and then $P=O_{k-2}(x)$ shows that $O_{k-2}(x)=O_{k-2}$. Thus, $H_{sur}(k)$ is true.\\

\noindent\textbf{Proof of $H_E(k)$.}

Let $k \geq 4$, and let $y_1,y_2\in V_{k}$ such that $N_{k-2}(y_1)=N_{k-2}(y_2)$. We will show that for all $p\in\int{0,k-2}\setminus\set{1}, N_p(y_1)=N_p(y_2)$.
From $H_{N}(k)$ applied to $y_1$ and $y_2$, we get:\\
$N_{k-2}(y_1)= W_{k-2}(y_1)=\set{t \in V_{k-2}\ |\ (O_1(t),\ldots ,O_{k-4}(t))=(O_1(y_1),\ldots ,O_{k-4}(y_1))$ $\text{ and } O_{k-3}(y_1)\subseteq O_{k-3}(t)\subseteq O_{k-2}(y_1)}$, and $N_{k-2} (y_2) = W_{k-2}(y_2)$.
Since $N_{k-2} (y_1)=N_{k-2}(y_2)$, then by considering a common element $t$ of these two sets (which are non empty since $O_{k-2}(y_1)\in\cur{O}$), we have $(O_1(y_1),\ldots,O_{k-4}(y_1))=(O_1(y_2),\ldots,O_{k-4}(y_2))$ (using the usual convention on empty sequences). We now prove that $O_{k-3}(y_1) = O_{k-3}(y_2)$. From recursion hypothesis $H_{sur}(k)$, we know that there exists $t_1\in V_{k-2}$ such that $S(t_1)=(O_1(y_1),\ldots ,O_{k-3}(y_1))$. This element $t_1$ is clearly an element of $N_{k-2}(y_1)$, and then is an element of $N_{k-2}(y_2)$. Then, we have $O_{k-3}(y_2) \subseteq O_{k-3}(t_1) = O_{k-3}(y_1)$. Symmetrically, by considering an element $t_2\in V_{k-2}$ such that $S(t_2)=(O_1(y_2),\ldots ,O_{k-3}(y_2))$, we obtain $O_{k-3}(y_1) \subseteq O_{k-3}(y_2)$. And finally, we have $O_{k-3}(y_1) = O_{k-3}(y_2)$.
Then, regardless of the value of $k\geq 4$, we have $N_0(y_1) = O_1(y_1) = O_1(y_2) = N_0(y_2)$, which is enough to prove $H_E(k)$ when $k=4$. Let us complete the general case where $k\geq 5$ by considering some $p\in\int{2,k-3}$ and showing that $N_p(y_1)=N_p(y_2)$. From recursion hypothesis $H_{N}(k)$, we have $N_p(y_1)= \set{t \in V_p\ |\ (O_1(t),\ldots ,O_{p-2}(t))=(O_1(y_1),\ldots ,O_{p-2}(y_1)) \text{ and }  O_{p-1}(y_1)\subseteq O_{p-1}(t)\subseteq O_p(y_1)}$. And since $p\leq k-3$, $(O_1(y_1),\ldots,O_{p}(y_1))=(O_1(y_2),\ldots,O_{p}(y_2))$, which implies $N_p(y_1) = N_p(y_2)$. Thus, for all $p\in\int{0,k-2}\setminus\set{1}, N_p(y_1)=N_p(y_2)$.

This shows that $H_E(k)$ is true, which ends the recursion step and the proof of Theorem~\ref{ThCharSeq}.

\end{proof}

The termination of the series directly follows from the bijection theorem (Theorem~\ref{ThCharSeq} above) between the vertices of the multipartite graph and the chains of $\cur{L}$.

\begin{theorem}[Termination theorem]\label{CFSstop}
For any graph $G$, the clean-factor series $(G_i)_{i\geq 1}$ generated by $G$ terminates.
\end{theorem}

\begin{proof}
Theorem~\ref{ThCharSeq} states that the characterising sequence $(O_1(x),\ldots ,O_{k-1}(x))$ of any node $x$ at level $k$ is such that $O_1(x)\subsetneq \ldots \subsetneq O_{k-1}(x)$. The strict inclusions imply that the length of the characterising  sequence, which is equal to $k-1$, cannot exceed $h+1$, where $h$ is the height of $\cur{L}$. Since $h\leq n-2$, necessarily $V_{n+1}$ is empty. It follows that the clean-factor series terminates and that the multipartite graph on which it terminates has at most $n+1$ levels, that is the upper level has index at most $n$.
\end{proof}

In our definition of the clean-factor series, the first bipartite graph of the series is always the vertex-clique-incidence bipartite graph of some graph. It is worth to note that our result of termination is actually more general: the iteration of the clean-factor operator starting from an arbitrary bipartite graph always terminates too.

\begin{corollary}\label{cor:bipgen}
For any bipartite graph $H$, the series of multipartite graphs obtained by iteratively applying the clean-factor operator starting from $H$ terminates.
\end{corollary}

\begin{proof}
Let $H=(V_0,V_1,E)$ be an arbitrary bipartite graph. Now, consider the vertex-clique-incidence bipartite graph $H'=(V'_0,V_1,E')$ built from $H$ in the following way: for each vertex $y\in V_1$, add a particularising vertex $x\in V'_0$ linked only to $y$. Then, in $H'$, the sets of neighbours on $V'_0$ of vertices of $V_1$ are not pairwise included and it follows that $H'$ is the vertex-clique-incidence bipartite graph of some graph\footnote{Graph $G'$ is simply the graph whose vertex set is $V'_0$ and whose maximal cliques are the neighbourhoods in $H'$ of the vertices of $V_1$.} $G'$. Moreover, since the vertices we added on level $V'_0$ are included in only one maximal clique of $G'$, then the vertices at level $V_2$ in the clean-factor graph of $H$ and $H'$ are the same and have the same neighbourhoods. And this holds for all other levels of the series as well: from level $V_1$ and above, the series of $H'$ is identical to the one of $H$. And since from Theorem~\ref{CFSstop}, the series of $H'$ terminates, so does the series of $H$.
\end{proof}

%
%
%

\section{Practical utility of the model}\label{sec-pract}

In addition to the theoretic questions we addressed, our work was motivated by designing a model of complex networks that, while remaining very general, encompass both the local density and the heterogeneous degree distribution of those graphs encountered in practice. In this section, we emphasise on the fact that our modelling object, the multipartite graph on which the clean factor series terminates, which we call the \emph{clean-factor decomposition}, is suitable for practical use, with regard to size and time of computation. This allowed us to compute the clean-factor decomposition of very large graphs having hundreds of thousands of vertices and millions of edges.
These practical results are not presented here since they are far beyond the scope of this work. But, in the following, we give theoretic evidence of why the clean-factor decomposition is a suitable model to manipulate large real-world instances of graphs, based on common properties of those graphs.

The size of the multipartite graph $M$ obtained at termination of the clean-factor series can be exponential in theory, as the number of maximal cliques itself may be exponential. But in practice, its size is quite reasonable and it can be computed efficiently. Indeed, the size of $M$ mainly depends on the complexity of imbrication of maximal cliques, namely on the number of chains of $\cur{L}$ (Theorem~\ref{ThCharSeq}). Theorem~\ref{th:size} below shows that under reasonable hypotheses, this number is linearly bounded and the size of $M$ only linearly depends on the number of vertices of $G$.

It must be clear that our hypotheses imply that the number of maximal cliques of $G$ is linearly bounded, as the vertices on level $V_1$ of $M$ are precisely the maximal cliques of $G$. But on the other hand, note that this bound on the number of maximal cliques is not sufficient to guarantee a polynomial bound on the number of vertices of $M$: there may still be an exponential number of vertices on the upper levels of $M$. Theorem~\ref{th:size} below shows that this does not happen under our hypotheses.

\begin{theorem}\label{th:size}
If every vertex of $G$ is involved in at most $k$ maximal cliques and if every maximal clique of $G$ contains at most $c$ vertices, then we have $$|V(M)|\leq min(k\,2^c\,c!\, , 2^k\,k!+1)\times n$$
\end{theorem}

\begin{proof}
Thanks to Theorem~\ref{ThCharSeq}, we obtain an upper bound on $|V(M)|$ by bounding the number of strictly increasing sequences of the form $(O_1,\ldots,O_i)$ such that $O_1,\ldots,O_{i-1}\in \cur{O}$.

First, we use the fact that all such sequences are sub-sequences of those obtained starting from a clique $O_i$ and recursively removing one vertex at each step until one obtains a pair $O_1$. The number of such sequences starting with a fixed clique is at most the number of orders on the $c$ vertices of the clique, that is $c!$. And the number of sub-sequences of a sequence of length $c$ is $2^c$. Finally, since each vertex is included in at most $k$ maximal cliques, the number of maximal cliques is at most $k\,n$. Then, there are at most $k\,2^c\,c!\,n$ increasing sequences made of elements of $\cur{O}$, which are in bijection with the vertices of $M$ of level at least $2$. Moreover, note that our counting also includes the sequences made of one single maximal clique of $G$ and the sequences made of one single set $\cur{O}$ which is a singleton. Those particular sequences are in bijection with the vertices of $M$ at level $1$ and $0$ respectively. Then, we obtain $|V(M)|\leq k\,2^c\,c!\,n$.

Clearly, since $\cur{O}\subseteq \cur{O}'$, the number of strictly increasing sequences made of elements of $\cur{O}$ is at most the number of strictly increasing sequences made of elements of $\cur{O}'$. Another way to count those latter sequences is to count those starting with a fixed minimal set $O_{min}\in\cur{O}'\setminus\set{\varnothing}$. Since $\cur{O}'$ is closed under intersection, minimal elements of $\cur{O}'\setminus\set{\varnothing}$ are pairwise disjoint and therefore their number is at most $n$. The sequences having $O_{min}$ as first set can be formed by starting from a clique containing $O_{min}$ and iteratively intersecting it with another clique containing $O_{min}$. By hypothesis, there are at most $k$ cliques containing a given $O_{min}$, and therefore $k!$ orders on these $k$ cliques. Each order gives rise to a sequence of elements of $\cur{O}'$, which contains $2^k$ sub-sequences. Thus, there are at most $k!\,2^k$ strictly decreasing sequences of elements of $\cur{O}'$ having $O_{min}$ as first element. Note that the counting we made actually also comprises the sequences made of one single maximal clique of $G$. Consequently, the number of vertices in $M$ at level at least $1$ is at most $k!\,2^k\,n$, and adding the $n$ vertices at level $0$ to this count we obtain the bound $|V(M)|\leq (2^k\,k!+1)\times n$, which completes the proof.
\end{proof}

In practice, parameters $k$ and $c$ are quite small, as they are often constrained by the context where the graphs come from (e.g. social networks, computer networks, citation networks) independently from the size of the graph. Then, the size of $M$ is reasonable in practice, namely $O(n)$ for class of graphs where $k$ and $c$ are bounded. An important consequence is that for those graphs it is possible to compute $M$ in low polynomial time. For example, under those hypotheses, the algorithm of \cite{WWDYC06} enumerates all maximal cliques of graph $G$ in linear time with regard to the number of maximal cliques, that is $O(n)$ time in this case. Moreover, \cite{GNS09} shows that, for general bipartite graphs on $|B|$ vertices, it is possible to enumerate their maximal bicliques in $O(|B|^2)$ time per biclique (see also \cite{AACFHS04} for a survey on maximal bicliques enumeration). In the computation of $M$, at any stage we have to compute the maximal bicliques of the bipartite graph between the uppermost level and the rest of the levels. Since, from Theorem~\ref{th:size}, under our hypotheses, the size of $M$ is $O(n)$, then the time needed to compute one maximal biclique is $O(n^2)$. And as we need to compute at most $O(n)$ bicliques along the algorithm, it follows that the total time spent by the algorithm for computing the maximal bicliques involved in the construction of $M$ is $O(n^3)$. Finally, as the rest of the treatments needed for the construction of $M$ can be achieved in polynomial time, it turns out that, under the hypotheses of Theorem~\ref{th:size}, $M$ can be computed in polynomial time.

These facts explain that, in practice, using as black boxes the implementation \cite{TTTprog} of \cite{TTT06}'s algorithm for enumeration of maximal cliques and the implementation \cite{LCM} of \cite{UAUA04}'s algorithm for enumeration of maximal bicliques, we could compute the clean-factor decomposition of graphs with thousands and even hundred of thousands of nodes. Indeed, we did so for a protein interaction network of $1\,458$ vertices and $1\,948$ edges, a movie actors network of $392\,340$ vertices and $15\,038\,083$ edges, and a piece of the world-wide-web graph of $325\,740$ nodes and $1\,090\,108$ edges (all can be found at \cite{CCNR}).

This shows that, even though the problem of computing the maximal cliques and bicliques is NP-hard for arbitrary graphs, for graphs encountered in practice, since this computation can be done in polynomial time (under the hypotheses of Theorem~\ref{th:size}), it is possible to efficiently compute the clean-factor series. This makes the clean-factor model a very promising tool for modelling complex networks.



\section{Conclusion}\label{sec-conclu}

In this paper, we studied the termination of the weak-factor operator, which is a multipartite graph operator appeared in the context of complex network modelling. One key issue in this context is that the series obtained by iteratively applying the operator terminates, as this is mandatory in order to obtain an object suitable for modelling. Since the weak-factor series does not always terminate, we designed a refinement of this operator, called the clean-factor graph, whose series terminates for all input graphs. And we showed that this modelling approach is practically efficient in the sense that the clean-factor series can be computed even for large graphs, under reasonable assumptions on their structure.

The first question arising from our work is to find minimal restrictions of the weak-factor operator that guarantee termination for all graphs. Indeed, it is crucial in practice to introduce constraints as light as possible, since those constraints, that have to be respected during the random generation process, makes this process more intricate to design and less efficient. In particular we ask whether the condition requiring equality of the neighbourhoods at level $V_{k-3}$ in the definition of the clean-factor graph can be replaced by a condition requiring only that these neighbourhoods share at least two common vertices.


Moreover, the use of multipartite graphs as models of complex networks, in the spirit of the bipartite decomposition~\cite{ipl04guillaume,physicaa06guillaume}, asks for some other important questions. In this context, the key issue is to generate a random multipartite graph while preserving the properties of the original graph. To do so, one has to express the properties to preserve as functions of basic multipartite properties (like degrees, for instance) and to generate random multipartite graphs satisfying these properties. This is a very promising direction for complex network modelling, but much remains to be done.

\medskip
\noindent
{\bf Acknowledgements.}
We warmly thank Thanh Qui Nguyen and The Hung Tran for helpful discussions on the subject, as well as Cl\'emence Magnien and St\'ephan Thomass\'e for their comments on the writing of the article.


\small{
\bibliographystyle{plain}
\bibliography{xbib,perso,clique-graphs}
}

\end{document}